\documentclass[journal]{IEEEtran}
\usepackage{amsthm}
\usepackage{amsmath}
\usepackage{amsfonts}
\usepackage{lineno, hyperref}
\usepackage{amssymb}
\usepackage{xcolor,graphicx,epstopdf,caption,subcaption}
\usepackage{enumitem}
\newtheorem{assumption}{Assumption}
\newtheorem{remark}{Remark}
\newtheorem{theorem}{Theorem}
\newtheorem{lemma}{Lemma}

\newcommand{\bgeq}{\begin{equation}}
\newcommand{\edeq}{\end{equation}}
\newcommand{\bgdm}{\begin{displaymath}}
\newcommand{\eddm}{\end{displaymath}}

\begin{document}

\title{Fault Detection and Isolation of Uncertain Nonlinear Parabolic PDE Systems\thanks{This work was supported in part by the National Science Foundation under Grant CMMI-1929729.}
}

\author{Jingting Zhang, Chengzhi Yuan, Wei Zeng, Cong Wang 
\thanks{J. Zhang and C. Yuan are both with the Department of Mechanical, Industrial and Systems Engineering, University of Rhode Island, Kingston, RI 02881, USA (e-mail: { jingting\_zhang@uri.edu; cyuan@uri.edu})}
\thanks{W. Zeng is with the School of Mechanical and Eletrical Engineering, Longyan University, Longyan 364012, China (e-mail: {zw0597@126.com})}
\thanks{C. Wang is with the School of Control Science and Engineering, Shandong University, Jinan 250061, China (e-mail: {wangcong@sdu.edu.cn})}
}


\maketitle

\begin{abstract}
This paper proposes a novel fault detection and isolation (FDI) scheme for distributed parameter systems modeled by a class of parabolic partial differential equations (PDEs) with nonlinear uncertain dynamics. 
A key feature of the proposed FDI scheme is its capability of dealing with the effects of system uncertainties for accurate FDI. 
Specifically, an approximate ordinary differential equation (ODE) system is first derived to capture the dominant dynamics of the original PDE system. 
An adaptive dynamics identification approach using radial basis function neural network is then proposed based on this ODE system, so as to achieve locally-accurate identification of the uncertain system dynamics under normal and faulty modes.  
A bank of FDI estimators with associated adaptive thresholds are finally designed for real-time FDI decision making. 
Rigorous analysis on the FDI performance in terms of fault detectability and isolatability is provided. 
Simulation study on a representative transport-reaction process is conducted to demonstrate the effectiveness and advantage of the proposed approach. 

\end{abstract}

\begin{IEEEkeywords}
Partial differential equations, fault detection and isolation, adaptive dynamics identification, deterministic learning, neural network, distributed parameter systems.
\end{IEEEkeywords}

\IEEEpeerreviewmaketitle

\section{Introduction}

\IEEEPARstart{D}{istributed} parameter systems (DPSs) are dynamical systems with inputs, outputs, and process parameters that may vary temporally and spatially \cite{SonHL.TCYB18, MonY.CNSNS12, XuFY.TCYB21}, which are usually modeled by partial differential equations (PDEs).  
Some typical examples include fluid flow process \cite{HerM.IS99}, biological process \cite{BouCK.SP15}, convection diffusion reaction process \cite{ElFG.AIChE07} and thermal process \cite{LuZH.TII16}. 
Particularly, due to the ever-increasing technical demands, fault diagnosis of DPSs has been an area of significantly growing interests. 
It is a critical step to realize fault tolerant operations for minimizing performance degradation and avoiding dangerous situations, such that safety and reliability of DPSs can be guaranteed.  
To this end, the past decades have witnessed tremendous progress in the research of fault diagnosis for DPSs, leading to a large variety of methods, see, e.g., \cite{RusCB.SSBM12, DemA.JRNC12, CaiFS.AUT16, FisD.AUT20, FenWW.TCYB21} and the references therein. 

As opposed to the substantially growing body of literature on fault detection (FD) of DPSs (e.g., \cite{RusCB.SSBM12, DemA.JRNC12, CaiFS.AUT16, DeyPM.AUT19, ZhaY.CSL20}), study on the fault isolation (FI) problem has gained quite limited success, especially for those DPSs with nonlinear unstructured uncertain dynamics. 
Some research efforts have been devoted to the development of FI methods for DPSs with precisely known models. For example, 
\cite{BanK.ACC12} proposed an FI scheme using a finite-dimensional geometric approach. 
In \cite{ElFG.AIChE07}, the FI problem for DPSs with various actuator faults has been investigated.
For the FI problem of DPSs with nonlinear uncertain dynamics, the research is still under-explored. 
One of the technical difficulties is that the dynamics of faults occurring in the system could be hidden within the system's general uncertain dynamics (e.g., unmodeled uncertainties), such that the fault feature could not be accurately identified for FI purpose. 
Some attempts have been made to overcome this difficulty. 
The FI scheme proposed in \cite{CaiJ.ACC16} is able to distinguish the effects between occurring fault and system uncertainties. 
\cite{GhaE.AUT09, GhaE.CDC07} developed a Lyapunov function-based FI scheme for DPSs, in which system uncertainties were handled by active control strategies. 
However, all these existing schemes have not appropriately dealt with the system uncertainties in the sense that occurring faults are typically required to have sufficiently large magnitudes (e.g., larger than those of the system uncertainties), limiting their wider applicability in practice.

To overcome the above deficiencies, a promising strategy is to realize  accurate modeling  of  the system uncertain dynamics. Adaptive neural network (NN)--as commonly used in the field of control and modeling of DPSs with uncertain dynamics (see, e.g., \cite{WuL.TNN08, QiL.CES09, ZhaTLJ.TNNLS16})--provides a powerful tool for this purpose. 
However, different from the control problem, where modeling errors of NN can typically be compensated by the controllers, using NN for accurate fault detection and isolation (FDI) of DPSs is rather challenging. 
This is because NN approximation errors often have negative impacts on the FDI residual signals, which cannot be structured and decoupled from the occurring fault, leading to possible misjudgment of FDI. 
To minimize such effects of NN approximation errors, a key technical challenge is to satisfy the so-called persistently exciting (PE) condition of associated NN regressor vectors \cite{YuaW.SCL11}. 
Recently, the deterministic learning (DL) theory proposed in \cite{dlt09wChD} has demonstrated that with radial basis function neural networks (RBF NN), almost any recurrent trajectory can result in the satisfaction of a partial PE condition. 
As a result, locally accurate RBF NN identification of nonlinear unstructured  uncertain dynamics can be achieved along the recurrent trajectory, and approximation errors of the associated NN can be guaranteed arbitrarily small \cite{WanH.TNN07}.  
With this important property, the DL theory has been recognized in recent years as a new and effective paradigm for the design of FDI schemes for general nonlinear uncertain systems, see \cite{CheW.IJACSP14, CheWH.TNNLS13, ZhaY.ACCESS19, ZhaYSHW.TCYB19}. 
This has opened new doors to the field though, existing DL-based FDI schemes still suffer from some limitations.
For example, the FI methods of \cite{CheW.IJACSP14, CheWH.TNNLS13, ZhaY.ACCESS19} require the dynamics of occurring faults to perfectly match those pre-defined/pre-trained faults.  
However, in many practical situations, e.g., when the system suffers disturbances resulting in changes of initial conditions or system parameters, the dynamics of occurring fault often exhibit some differences from those of the matched fault, which thus could result in missed/false alarm under the FI schemes of \cite{CheW.IJACSP14, CheWH.TNNLS13, ZhaY.ACCESS19}. 
Furthermore, virtually all of these existing FDI methods are focused on lumped parameter systems (LPSs) modeled by finite-dimensional ordinary differential equations (ODEs), which cannot be directly applied for infinite-dimensional DPSs as considered in the current paper. 

To extend the DL theory to the FDI problem of DPSs, one promising strategy is to use a finite set of ODEs to approximate the PDE model based on model reduction methods \cite{LiQ.JPC10}. 
Conventional spatial discretization-based approaches often lead to a high-order ODE system \cite{BenOCW.SIAM17}, which however could  be computationally expensive for real-time implementation. 
An alternative approach is based on the Galerkin method \cite{MonY.CNSNS12, WuL.TNN08, ChrD.JMAA97}. 
The key idea of the Galerkin method is as follows. It is known that for dissipative parabolic DPSs, the eigenspectrum of the associated spatial differential operator can be partitioned into a finite set of slow eigenvalues and an infinite set of fast but stable eigenvalues \cite{DenLC.TCST05}. % slow set and an infinite stable fast complement. 
By neglecting the fast stable components, a low-order ODE system can be obtained to approximate the dominant dynamics of the PDE system, which then could be utilized to facilitate the subsequent FDI design.   
It is worth mentioning that for FDI of DPSs, only a few research results have been obtained in \cite{BanK.ACC12, ElFG.AIChE07, GhaE.AUT09, GhaE.CDC07}, which unfortunately are applicable only to some special types of faults such as actuator faults, but cannot address the FDI problem for DPSs with more general faults. 

In this paper, we aim to investigate effective detection and isolation approaches for general faults occurring in DPSs modeled by a class of parabolic PDEs with nonlinear unstructured uncertain dynamics. A novel Galerkin-DL-based FDI scheme will be proposed. 
Specifically, with the Galerkin method, an approximate ODE model is first derived to capture the dominant dynamics of the PDE system. A DL-based adaptive dynamics identification approach is then developed based on this ODE system to realize  locally-accurate identification of the uncertain system dynamics under normal and all faulty modes. The associated knowledge can be obtained and stored in constant RBF NN models.
Afterwards, a bank of FDI estimators are designed with these constant models, where the FD estimators are used to detect the occurrence of a fault, and the FI estimators (which will be activated once the occurring fault is detected) are used to identify the type of the occurring fault.  
Their generated residuals can be used to characterize the dynamics of the occurring fault and distinguish it from the system uncertain dynamics for accurate FDI. 
Adaptive thresholds associated with such FDI residuals are further designed to facilitate real-time FDI decision making. 
In particular, to address the aforementioned robustness issue encountered by the FI methods of \cite{CheW.IJACSP14, CheWH.TNNLS13, ZhaY.ACCESS19}, novel adaptive thresholds instead of fixed/constant thresholds are designed in this paper, % as adopted in \cite{BanK.ACC12, ElFG.AIChE07, CaiJ.ACC16, GhaE.AUT09}. 
such that successful isolation can still be guaranteed even when the occurring fault does not exactly match any of the pre-trained faults. As such, the proposed  FI scheme possesses improved robustness against slight deviations of fault dynamics due possibly to unexpected system changes in, for example, initial conditions and system parameters, as discussed above.  
We stress that our FDI scheme does not require the faults to be of any special type (e.g., actuator faults as required in \cite{GhaE.AUT09, GhaE.CDC07, BanK.ACC12} and/or sensor faults as required in \cite{CaiJ.ACC16}), but is applicable to general system faults. 
Rigorous analysis on the FDI performance is conducted to demonstrate that our approach develops better fault detectability and isolatability compared to the existing methods of \cite{CaiJ.ACC16, ElFG.AIChE07, ArmD.AIChE08, DemA.JRNC12}. 
Moreover, extensive simulations applied to a representative transport-reaction process are also conducted to demonstrate the effectiveness and advantage of the proposed new methodologies. 

It should be pointed out that this research work significantly expands our previous work \cite{ZhaY.CSL20} which was focused on only the FD problem of uncertain parabolic PDE systems; while in this paper, we consider both the FD and FI problems. 
In addition, new adaptive thresholds are proposed for more accurate and efficient FDI, which advance the fixed/constant thresholds proposed in \cite{ZhaY.CSL20}. 
%%
%different from the FD scheme in \cite{ZhaY.CSL20}, whose design does not need to precisely know the information of occurring faults, the design of FI scheme requires the dynamics/model of each class of possible occurring faults to be available. This is quite challenging especially for the DPSs with uncertain system dynamics.  
%%
%To overcome this difficulty, in this paper, the system uncertain dynamics under each faulty mode is accurately identified, and a bank of FI estimators are designed with the learned knowledge, which can be used to represent/reconstruct the dynamics of each class of faults. 
%%
%Adaptive thresholds associated with these FI estimators are further designed to specify the similarity measure of each fault class, which can facilitate to accurately separate the dynamics of faults belonging to different classes. This is an important feature of our FI scheme different from the fixed/constant threshold used in the FD scheme of \cite{ZhaY.CSL20}. 
%
Furthermore, this paper also provides rigorous analysis to characterize the properties of the proposed FDI scheme, which include: (i) fault detectability conditions characterizing the class of faults that can be detected, and (ii) fault isolatability conditions characterizing the class of faults that can be isolated.

The contributions of this paper are summarized as follows: 
(i) The FDI problem for uncertain nonlinear parabolic PDE systems is addressed using a novel Galerkin-DL-based adaptive dynamics identification approach, which can achieve locally-accurate identification of the dominant uncertain dynamics of the PDE system; 
(ii) New adaptive-threshold-based FDI decision making schemes are proposed, which are capable of dealing with general faults occurring in parabolic PDE systems, including those faults that do not exactly match the pre-defined/pre-trained faults; 
(iii) Rigorous analysis on FDI performance, including fault detectability and isolatability conditions, is provided to demonstrate the effectiveness of the proposed  approaches.

The remainder of this paper is organized as follows. Some preliminaries and the problem statement are provided in Section \ref{sec_preli-problem}. The DL-based adaptive dynamics learning approach is presented in Section \ref{sec_identification}. The FD scheme is proposed in Section \ref{sec_FD}, and the FI scheme is given in Section \ref{sec_FI}. Simulation studies are presented in Section \ref{sec_simulation}. Section \ref{sec_conclusion} concludes the paper. 

\noindent \textbf{Notation.} $\mathbb{R}$, $\mathbb{R}_+$ and $\mathbb{N}_+$ denote, respectively, the set of real numbers, the set of positive real numbers and the set of positive integers; $\mathbb{R}^{m \times n}$ denotes the set of $m \times n$ real matrices; $\mathbb{R}^n$ denotes the set of $n \times 1$ real column vectors; %$I_n$ denotes the $n \times n$ identity matrix; the open ball $B_r=\{x\in \mathbb{R}^n:\,\left\|x\right\|<r\}$ with $r$ being an arbitrary positive constant; 
$|\cdot|$ is the absolute value of a real number; $\left\|\cdot\right\|$ is the 2-norm of a vector or a matrix, i.e. $\left\|x\right\|=(x^\top x)^{\frac{1}{2}}$.  %$\left\|\cdot\right\|_1$ is the $\mathcal{L}_1$-norm of a vector or a matrix, i.e. $\left\|x(k)\right\|_1=\frac{1}{K}\sum_{h=k-K}^{k-1}|x(h)|$ $(k \geq K>1)$. %; $\lceil a \rceil$ denotes the least integer greater than or equal to a real number $a$.  
%$\delta(\epsilon)=O(\epsilon)$ denotes that if there exist positive constants $k_1$ and $k_2$ such that $|\delta(\epsilon)| \leq k_1|\epsilon|$, $\forall |\epsilon| \leq k_2$.

\section{Preliminaries and Problem Formulation} \label{sec_preli-problem}
\subsection{Radial Basis Function Neural Networks} \label{sec_RBFNN}

The RBF NNs can be described by $f_{nn}(Z)=\sum_{i=1}^{N_n}w_is_i(Z)=W^\top S(Z)$ \cite{tto92pM}, where $Z\in \Omega_Z \subset \mathbb{R}^q$ is the input vector, $W=[w_1,\cdots,w_{N_n}]^\top \in \mathbb{R}^{N_n}$ is the weight vector, $N_n$ is the NN node number, and $S(Z)=[s_1(\|Z-\varsigma_1\|),\cdots,s_{N_n}(\|Z-\varsigma_{N_n}\|)]^\top $, with $s_i(\cdot)$ being a radial basis function, and $\varsigma_i$ $(i=1,2,\cdots,N_n)$ being distinct points in state space. The Gaussian function $s_i(\|Z-\varsigma_i\|)=\exp[\frac{-(Z-\varsigma_i)^\top (Z-\varsigma_i)}{\nu_i^2}]$ is one of the most commonly used radial basis functions, where $\varsigma_i=[\varsigma_{i1},\varsigma_{i2},\cdots,\varsigma_{iq}]^\top $ is the center of the receptive field and $\nu_i$ is the width of the receptive field. The Gaussian function belongs to the class of localized RBFs in the sense that $s_i(\|Z-\varsigma_i\|)\rightarrow 0$ as $\|Z\|\rightarrow \infty$. It is noted that $S(Z)$ is bounded, i.e., there exists a real constant $S_M \in \mathbb{R}_+$ such that $ \|S(Z)\|\leqslant S_M$ \cite[Lemma 2.1]{dlt09wChD}.
It has been shown in \cite{tto92pM} that for any continuous function $f(Z):\Omega_Z \rightarrow \mathbb{R}$ where $\Omega_Z\subset \mathbb{R}^q$ is a compact set, and for the NN approximator, where the node number $N_n$ is sufficiently large, there exists an ideal constant weight vector $W^*$, such that for any $\epsilon^*>0$, $f(Z)=W^{*\top}S(Z)+\epsilon,\,\forall Z\in \Omega_Z$, where $|\epsilon|< \epsilon^*$ is the ideal  approximation error. The ideal weight vector $W^*$ is an ``artificial'' quantity used for analysis, which is defined as the value of $W$ that minimizes $|\epsilon|$ for all $Z\in\Omega_Z\subset\mathbb{R}^q$, i.e., $W^*\triangleq arg\textup{min}_{W\in\mathbb{R}^{N_n}}\{\textup{sup}_{Z\in\Omega_Z}|f(Z)-W^\top S(Z)|\}$. Moreover, based on the localization property of RBF NNs \cite{dlt09wChD}, for any bounded trajectory $Z(t)$ within the compact set $\Omega_Z$, $f(Z)$ can be approximated by using a finite number of neurons located in a local region along the trajectory: $f(Z)=W^{*\top}_\zeta S_\zeta(Z)+\epsilon_\zeta$, where $\epsilon_\zeta$ is the approximation error, with $\epsilon_\zeta=O(\epsilon)=O(\epsilon^*)$, $S_\zeta(Z)=[s_{j1}(Z),\cdots,s_{j\zeta}(Z)]^\top \in\mathbb{R}^{N_\zeta}$, $W_\zeta^*=[w^*_{j1},\cdots,w^*_{j\zeta}]^\top \in\mathbb{R}^{N_\zeta}$, $N_\zeta<N_n$, and the integers $j_i=j_1,\cdots,j_\zeta$ are defined such that $|s_{j_i}(Z_p)|>\theta$ ($\theta>0$ is a small positive constant) for some $Z_p\in Z(t)$.
In addition, it is shown in \cite{dlt09wChD} that for a localized RBF network $W^\top S(Z)$ whose centers are placed on a regular lattice, almost any recurrent trajectory\footnote{A recurrent trajectory represents a large set of periodic and periodic-like trajectories generated from linear/nonlinear dynamics systems. A detailed characterization of recurrent trajectories can be found in \cite{dlt09wChD}.} $Z(t)$ can lead to the satisfaction of the PE condition of the regressor subvector $S_\zeta(Z)$. 
This result is summarized in the following lemma.
\begin{lemma}[\cite{dlt09wChD}] \label{lem_DL}
Consider any recurrent trajectory $Z(t)$ that remains in a bounded compact set $\Omega_Z\subset \mathbb{R}^q$. For RBF network $W^\top S(Z)$ with centers placed on a regular lattice (large enough to cover compact set $\Omega_Z$), the regressor subvector $S_\zeta(Z)$ consisting of RBFs with centers located in a small neighborhood of $Z(t)$ satisfies the PE condition. %, i.e., there exist positive constants $\alpha_1,\alpha_2$ and $T_0$ such that $\alpha_1 I \leq \int_{t_{\tau}}^{t_{\tau}-T_0} S_\zeta(Z(\tau)) S_\zeta(Z(\tau))^\top  d \tau \leq \alpha_2 I$, $\forall t_{\tau} \geq 0$, where $I \in \mathbb{R}^{N_\zeta \times N_\zeta}$ is the identity matrix. 
\end{lemma}

\subsection{Problem Formulation} \label{sec_problem}

Consider a class of nonlinear parabolic PDE systems in one spatial dimension with a state-space description in the form of: 
\bgeq \label{PDE_sys}
\begin{aligned}
	\frac{\partial x(z,t)}{\partial t} =&\, a_1\frac{\partial x(z,t)}{\partial z}+a_2\frac{\partial^2 x(z,t)}{\partial z^2}+ f(x,u) \\
	&\, + \beta(t-t_0)\phi^k(x,u), 
	%\frac{\partial x}{\partial t} = a_1\frac{\partial x}{\partial z}+a_2\frac{\partial^2 x}{\partial z^2}+ f(x,u) + \beta(t-t_0)\phi^k(x,u),
\end{aligned}
\edeq
subject to the following boundary conditions and initial condition:
\bgeq \label{PDE_sys-bound-initial}
\begin{aligned}
	& m_i x(z_i,t) + n_i \frac{\partial x}{\partial z} (z_i,t) = d_i, \quad i=1,2 \\
%	& m_2 x(z_2,t) + n_2 \frac{\partial x}{\partial z} (z_2,t) = d_2, \\
	& x(z,0) = x_0(z),
\end{aligned}
\edeq
where $x(z,t) \in \mathbb{R}$ is system state; $u \in \mathbb{R}^q$ is system input; $z \in [z_1,z_2] $ is the spatial coordinate; $ t \in [0,\infty)$ is the time; $f(x,u) \in \mathbb{R}$ and $\phi^k(x,u) \in \mathbb{R}$ are unknown nonlinear functions satisfying locally Lipschitz continuous, which represent nonlinear uncertain system dynamics and deviations in system dynamics due to fault $k \in \{1,2,\cdots,N\}$ ($N \in \mathbb{N}_+$), respectively; $\beta(t-t_0)$ is the time profile of the occurring fault, with $\beta(t-t_0)=0$ for $t < t_0$ and $\beta(t-t_0)=1$ for $t \geq t_0$; $t_0$ is the unknown fault occurrence instant; $\frac{\partial x}{\partial z}$ and $\frac{\partial^2 x}{\partial z^2}$ are the first-order and second-order spatial derivatives of $x(z,t)$, respectively;  $a_1,a_2,m_1,m_2,n_1,n_2,d_1,d_2$ are known constants. 
In this paper, it is assumed that the system state $x(z,t)$ is measurable at all locations $z \in [z_1,z_2]$ for all time $t \in [0,\infty)$.

\begin{assumption} \label{ASS_recurrent}
For the PDE system (\ref{PDE_sys})--(\ref{PDE_sys-bound-initial}), the system input $u(t)$ and state $x(z_0,t)$ at any spatial point $z_0 \in [z_1,z_2]$ are bounded and recurrent for all $t \in [0,\infty)$.  
\end{assumption}

Denote $\mathcal{H} $ as a Hilbert space of 1-D functions defined on $[z_1,z_2]$ that satisfies the boundary conditions (\ref{PDE_sys-bound-initial}), with inner product and norm: $\left\langle \zeta_1, \zeta_2 \right\rangle = \int_{z_1}^{z_2} \zeta_1(z) \zeta_2(z) dz,\, \left\| \zeta_1 \right\|_2 = \left\langle \zeta_1,\zeta_1 \right\rangle ^{\frac{1}{2}} $, where $\zeta_1(z),\,\zeta_2(z)$ are two elements of $\mathcal{H}$. % and $\delta(z)$ is an appropriate weighting function. 
According to \cite{ElFG.AIChE07, WuL.TNN08}, the PDE system (\ref{PDE_sys})--(\ref{PDE_sys-bound-initial}) can be formulated as an infinite-dimensional system: 
\bgeq \label{ODE_sys-infinite}
\dot{\chi} = \mathcal{A}\chi + f(\chi,u) + \beta(t-t_0) \phi^k(\chi,u), \quad \chi(0) = \chi_0,
\edeq
where $\chi(t)=x(z,t)$ is the state function defined in $\mathcal{H}$, %$\chi_0 = x_0(z)$, 
and $\mathcal{A}$ is a differential operator in $\mathcal{H}$ defined as $\mathcal{A}x=a_1\frac{\partial x}{\partial z} + a_2\frac{\partial^2 x}{\partial z^2}, \, x\in D(\mathcal{A}):= \{x \in \mathcal{H} \,|\, \mathcal{A} x \in \mathcal{H},\, m_i x(z_i,t) + n_i \frac{\partial x}{\partial z} (z_i,t) = d_i, \, i=1,2\} $.
%\bgeq \label{PDE_operator}
%\begin{aligned}
%	\mathcal{A}x=a_1\frac{\partial x}{\partial z} +& a_2\frac{\partial^2 x}{\partial z^2}, \\
%	x\in D(\mathcal{A}):=& \{x \in \mathcal{H} \,|\, \mathcal{A} x \in \mathcal{H},\, \\
%	& m_i x(z_i,t) + n_i \frac{\partial x}{\partial z} (z_i,t) = d_i, \, i=1,2\}. 
%\end{aligned}
%\edeq
%
For the operator $\mathcal{A}$, the eigenvalue problem is defined as $\mathcal{A} \varphi_j = \lambda_j \varphi_j$ ($j=1,2,\cdots, \infty$), where $\lambda_j$ denotes an eigenvalue, and $\varphi_j$ denotes an eigenfunction. The eigenspectrum of $\mathcal{A}$, denoted by $\sigma(\mathcal{A})$, is defined as the set of all eigenvalues of $\mathcal{A}$, i.e., $\sigma(\mathcal{A})=\{\lambda_1,\lambda_2,\cdots,\lambda_{\infty}\}$. According to \cite{ElFG.AIChE07, WuL.TNN08}, for highly-dissipative PDE systems, the eigenspectrum of $\mathcal{A}$ can be partitioned into a finite-dimensional part consisting of $m$ ($m \in \mathbb{N}_+$) slow eigenvalues and a stable infinite-dimensional complement containing the remaining fast eigenvalues, and the separation between the slow and fast eigenvalues of $\mathcal{A}$ is large. These properties can be  satisfied by the majority of diffusion-convection-reaction processes \cite{ElFG.AIChE07}, and are formalized in the following assumption.  

\begin{assumption} \label{ASS_eigenspectrum}
(i) 
$\textup{Re}\{\lambda_1\} \geq \textup{Re}\{\lambda_2\} \geq \cdots \geq \textup{Re}\{\lambda_j\} \geq \cdots$, where $\textup{Re}\{\lambda_j\}$ denotes the real part of $\lambda_j$; 
(ii) 
$\sigma(\mathcal{A})$ can be partitioned as $\sigma(\mathcal{A})=\sigma_s(\mathcal{A})+\sigma_f(\mathcal{A})$, where $\sigma_s(\mathcal{A})$ consists of the first $m$ number of eigenvalues, that is, $\sigma_s(\mathcal{A})=\{\lambda_1,\lambda_2,\cdots,\lambda_m\}$, and $\left|\frac{\textup{Re}\{\lambda_1\}}{\textup{Re}\{\lambda_{m}\}} \right|=O(1)$;
(iii) 
$\textup{Re}\{\lambda_{m+1}\} < 0 $ and $\left|\frac{\textup{Re}\{\lambda_m\}}{\textup{Re}\{\lambda_{m+1}\}}\right|=O(\iota)$, where $\iota:= \left|\frac{\textup{Re}\{\lambda_1\}}{\textup{Re}\{\lambda_{m+1}\}} \right|<1$ is a small positive constant.

\end{assumption}

Based on this assumption, consider the decomposition $\mathcal{H}=\mathcal{H}_s \oplus \mathcal{H}_f$, in which $\mathcal{H}_s=\textup{span}\{\varphi_1, \varphi_2,\cdots,\varphi_m \}$ denotes the finite dimensional space spanned by the slow eigenfunctions corresponding to $\sigma_s(\mathcal{A})$, and $\mathcal{H}_f=\textup{span}\{\varphi_{m+1}, \varphi_{m+2},\cdots, \varphi_{\infty} \}$ denotes the infinite dimensional complement one spanned by the fast eigenfunctions corresponding to $\sigma_f(\mathcal{A})$. 
Under such a decomposition and through separation of time and spatial variables \cite{ElFG.AIChE07, WuL.TNN08}, the PDE system (\ref{ODE_sys-infinite}) can be rewritten in the following equivalent form:
\bgeq \label{ODE_sys-sf0}
\begin{aligned}
	&\dot{x}_s = A_s x_s + f_s(x_s,\chi_f,u) + \beta(t-t_0) \phi_s^k(x_s,\chi_f,u),\\
	&\dot{\chi}_f = A_f \chi_f + f_f(x_s,\chi_f,u) + \beta(t-t_0) \phi_f^k(x_s,\chi_f,u),\\
	&x_s(0) = x_{s_0}, \, \chi_f(0)=\chi_{f_0},
\end{aligned}
\edeq
where $ x_s = [x_{s_1}, \chi_{s_{2}}, \cdots,x_{s_m}]^\top \in \mathbb{R}^m$, $\chi_f = [\chi_{f_{m+1}},\cdots,$\\$\chi_{f_{\infty}}]^\top \in \mathbb{R}^\infty $, $ A_s = \textup{diag}\{\lambda_1,\cdots,\lambda_m\}$, $f_s=\left\langle \varphi_s, f \right\rangle$, $\phi_s^k = \left\langle \varphi_s, \phi^k \right\rangle$, $x_{s_0} = \left\langle \varphi_s, \chi_0 \right\rangle$, $ A_f = \textup{diag}\{\lambda_{m+1},\cdots,\lambda_{\infty}\}$, $f_f=\left\langle \varphi_f, f \right\rangle$, $\phi_f^k = \left\langle \varphi_f, \phi^k \right\rangle$, $\chi_{f_0} = \left\langle \varphi_f, \chi_0 \right\rangle$ with $\varphi_s = [\varphi_1, \cdots,\varphi_m]^\top$ and  $\varphi_f = [\varphi_{m+1}, \cdots,\varphi_\infty]^\top$. 
%
%Particularly, under Assumption \ref{ASS_eigenspectrum}, $A_s$ is a diagonal matrix of dimension $m \times m$ of the form $ A_s = \textup{diag}\{\lambda_1,\cdots,\lambda_m\}$, and $A_f$ is an unbounded differential operator of the form $ A_f = \textup{diag}\{\lambda_{m+1},\cdots,\lambda_{\infty}\}$ which contains all negative eigenvalues. %, as analyzed in \cite{ElFG.AIChE07}. 
%
By neglecting the fast modes, we can obtain the following finite-dimensional ODE model  to characterize the dominant dynamics of the PDE system in (\ref{ODE_sys-infinite}):   
\bgeq \label{ODE_sys} 
	\dot{x}_s = A_s x_s + f_s(x_s,u) + \beta(t-t_0) \phi_s^k(x_s,u), \, x_s(0) = x_{s_0}. 
\edeq

\begin{remark}
The process of model-reduction in Eqs. (\ref{ODE_sys-infinite})--(\ref{ODE_sys}) is based on the Galerkin method, as adopted in existing works \cite{ElFG.AIChE07, WuL.TNN08}, which is included here for the completeness of presentation. 
Note that such a process is not the major contribution of this paper, thus its thorough analysis is not provided here. Interested readers are referred to \cite{DenLC.TCST05, ElFG.AIChE07, WuL.TNN08} for more details. 
\end{remark}

%\begin{remark} \label{rem_operator}
%The operators $\mathcal{P}_s$ and $\mathcal{P}_f$ can be specifically defined as $ \mathcal{P}_s \psi = \left\langle \psi,\varphi_s(z) \right\rangle$ and $ \mathcal{P}_f \psi = \left\langle \psi,\varphi_f(z) \right\rangle$, respectively, where $\psi \in \mathcal{H}$, $\varphi_s(z) = [\varphi_1(z),\cdots,\varphi_m(z)]^\top $, $ \varphi_f(z) = [\varphi_{m+1}(z),\cdots,$ $\varphi_{\infty}(z)]^\top $, as adopted in \cite{WuL.TNN08, ElFG.AIChE07}. 
%%
%%Based on this, each element of (\ref{ODE_sys-sf0}) can be calculated as: $x_{s}(t) = \left\langle x(z,t),\varphi_s(z) \right\rangle$, $\chi_{f}(t) = \left\langle x(z,t),\varphi_f(z) \right\rangle $, $f_{s}(x_s,\chi_f,u)= \left\langle f(x,u),\varphi_s(z) \right\rangle$, $ f_{f}(x_s,\chi_f,u)= \left\langle f(x,u),\varphi_f(z) \right\rangle $, $\phi_{s}(x_s,\chi_f,u)= \left\langle \phi(x,u),\varphi_s(z) \right\rangle$, $ \phi_{f}(x_s,\chi_f,u)= \left\langle \phi(x,u),\varphi_f(z) \right\rangle $, $ x_{s_0} = \left\langle x_0(z),\varphi_s(z) \right\rangle$, and $ \chi_{f_0} = \left\langle x_0(z),\varphi_f(z) \right\rangle $.
%\end{remark}

In the next sections, a novel FDI scheme will be proposed based on the ODE system (\ref{ODE_sys}), so as to achieve accurate detection and isolation for the occurring fault $\phi_s^k(x_s,u)$. 
Note that since the functions $f_s(x_s,u)$ and $\phi_s^k(x_s,u)$ are both unknown, the model (\ref{ODE_sys}) cannot be directly used for the FDI design. 
In view of this, the FDI scheme proposed in this paper will consist of three components: 
(i) adaptive dynamics learning, to achieve locally-accurate identification of the uncertain dynamics $f_s(x_s,u)$ and $\phi_s^k(x_s,u)$ in system (\ref{ODE_sys}) under the normal mode and all faulty modes; 
(ii) FD scheme, to achieve rapid detection of fault occurrence; 
and (iii) FI scheme, to realize accurate fault isolation, which will be activated once the occurring fault is detected. 
%

%{\JT
%\begin{remark}
%Note that our FDI schemes are developed for a class of parabolic PDE systems (\ref{PDE_sys}), which are currently not applicable to other classes of PDE systems, e.g., hyperbolic PDE systems. 
%%
%The main reason lies in that Assumption \ref{ASS_eigenspectrum} cannot be satisfied by hyperbolic PDE systems, such that the associated model-reduction method in Section \ref{sec_problem} cannot be applicable. 
%%
%Thus, to extend our FDI scheme to other classes of PDE systems, model-reduction for PDE systems can be achieved by employing interpolation methods instead of Galerkin method. Promising approaches along this direction for future work can be found in \cite{DonSW.ACCESS19, DonW.JBC15, ChaC.TFS10}. 
%
%\end{remark}
%}

\section{Identification of System Uncertain Dynamics} \label{sec_identification}

In this section, a DL-based adaptive dynamics learning approach will be developed to achieve accurate identification of the uncertain dynamics $f_{s}(x_s,u)$ and $\phi^k_{s}(x_s,u)$ in system (\ref{ODE_sys}) under all normal and faulty modes. 
%{\JT
%This scheme is similar to the one employed in our previous work \cite{YuaW.SCL11, ZhaY.CSL20}. It is included here for the completeness of presentation. 
%}

Consider the following faulty dynamic systems: 
\bgeq \label{Identify_sys0}
	\dot{x}_{s} = \, A_s x_{s} + f_{s}(x_s,u)+\phi^k_{s}(x_s,u), 
\edeq
where $k =0,1,\cdots, N$ denotes the $k$-th faulty mode, with $k=0$ representing the normal mode, i.e., $\phi_{s}^0(x_s,u)\equiv 0$. 
Since the system uncertainty $f_{s}(x_s,u)$ and occurring fault $\phi^k_{s}(x_s,u)$ in (\ref{Identify_sys0}) cannot be decoupled, by considering them together and defining a general fault function $\eta^k(x_s,u):=f_{s}(x_s,u)+\phi^k_{s}(x_s,u)$, we can rewrite the system (\ref{Identify_sys0}) as:
\bgeq \label{Identify_sys}
\begin{aligned}
	\dot{x}_{s_i} = \, \lambda_i x_{s_i} + \eta^k_i(x_s,u), \quad i=1,2,\cdots,m. 
\end{aligned}
\edeq 
%for all $i=1,\cdots,m$ and $k=0,1,\cdots,N$.

For the unknown function $\eta^k_i(x_s,u)$ in (\ref{Identify_sys}), according to the RBF NN approximation theory as presented in Section \ref{sec_RBFNN}, we know that there exists an ideal constant NN weight vector  $W_i^{k*} \in \mathbb{R}^{N_n }$ (with $N_n$ denoting the number of NN nodes) such that
\bgeq \label{Identify_funcapp}
	\eta_i^k(x_s,u)=W_i^{k*\top}S(x_s,u)+\varepsilon_{i_0}^k, % \quad i=1,\cdots,m
\edeq
%for all $i=1,\cdots,m$ and $k=0,1,\cdots,N$, 
where $S(x_s,u):\, \mathbb{R}^{m} \times \mathbb{R}^q \rightarrow \mathbb{R}^{N_n}$ is a smooth RBF vector and $\varepsilon_{i_0}^k$ is the estimation error satisfying $|\varepsilon_{i_0}^k|<\varepsilon_i^{*}$ with $\varepsilon_i^{*}$ being a positive constant that can be made arbitrarily small given a sufficiently large number of neurons.
Based on this, an adaptive dynamics identifier can be constructed:
\bgeq \label{Identify_obs}
\begin{aligned}
	\dot{\hat{x}}_i =&\, -a_{i}(\hat{x}_i-x_{s_i})+\lambda_i x_{s_i}+\hat{W}_i^{k\top}S(x_s,u),  \\
	\dot{\hat{W}}_i^k =&\, -\sigma_{i} \Gamma_{i} \hat{W}_i^k - \Gamma_{i} (\hat{x}_i-x_{s_i})  S(x_s,u),
\end{aligned}
\edeq
for all $i=1,2,\cdots,m$ and $k=0,1,\cdots,N$, 
where $\hat{x}_i $ is the identifier state, $x_{s_i}$ is the state of system (\ref{Identify_sys}), $\hat{W}_i^k \in \mathbb{R}^{N_n}$ is the estimate of $W_i^{k*}$ in (\ref{Identify_funcapp}), $a_i>0$, $\Gamma_{i} = \Gamma_{i}^\top  >0$, $\sigma_{i} >0$ are design constants with $\sigma_{i}$ being a small number. 

\begin{theorem} \label{theo_identify}
Consider the adaptive learning system consisting of the plant (\ref{Identify_sys}) and the identifier (\ref{Identify_obs}). Under Assumption \ref{ASS_recurrent}, with initial condition $\hat{W}_i^k(0)=0$, for all $i=1,\cdots,m$ and $k=0,1,\cdots,N$, we have: (i) all signals in the system remain bounded; (ii) the estimation error $\left| \hat{x}_i-x_{s_i} \right|$ converges to a small neighborhood around the origin; %and the NN weight $\hat{W}_i^k$ will converge to a small neighborhood of its optimal value $W_i^{k*}$; 
and (iii) a locally-accurate approximation of the unknown function $\eta_{i}^k(x_s,u)$ is achieved by $\hat{W}_i^{k\top}S(x_s,u)$ as well as $\bar{W}_i^{k\top}S(x_s,u)$ along the recurrent system trajectory $(x_s,u)$, where $\bar{W}_i^k := \frac{1}{t_2-t_1} \int_{t_1}^{t_2} \hat{W}_i^k(\tau) d \tau $ with $[t_1,\, t_2]$ representing a time segment after the transient process.
\end{theorem}

Detailed proof can be completed by following a similar line of the proof of \cite[Th. 3.1]{dlt09wChD}, thus is omitted here.

\begin{remark} \label{rem_identi-xs}
Implementing (\ref{Identify_obs}) requires information of the system state $x_s(t)$, which can be obtained by measuring the state signal $x(z,t)$ from  the original PDE system (\ref{PDE_sys})--(\ref{PDE_sys-bound-initial}) via $x_s(t) = \left\langle \varphi_s(z), x(z,t) \right\rangle$.

\end{remark}

Through the above learning process, the knowledge of unknown function $\eta_{i}^k(x_s,u)$ of (\ref{Identify_sys}) can finally be obtained and stored in the constant RBF NN model $\bar{W}_i^{k\top}S(x_s,u)$, i.e., 
\bgeq \label{Identify_funcapp-bar}
	\eta_{i}^k(x_s,u)=\bar{W}_i^{k\top}S(x_s,u)+\varepsilon_{i}^k, 
\edeq
for all $i=1,2,\cdots,m$ and $k=0,1,\cdots,N$, where the approximation error $\varepsilon_{i}^k$ satisfies $|\varepsilon_{i}^k|=O(\varepsilon_i^{*})<\xi^{*}_i$, with $\xi_i^{*}$ being a positive constant that can be made arbitrarily small by constructing a sufficiently large number of neurons \cite{YuaW.SCL11}. 
%

%\begin{remark} \label{rem_compare-iden}
%
%In \cite{DonSW.ACCESS19, DonW.JBC15}, an identification scheme for uncertain DPSs was developed by using the combination of DL and a cubic spline interpolation method. This scheme involves a high-order ODE model demanding a large number of RBF NN models for system identification. 
%%
%As a result, it may not be applicable for the design of FDI scheme, due to its heavy computational burdens. 
%%
%To overcome this issue, our approach is developed based on the Galerkin method, such that reduced number (i.e., $m$) of NN models can be used for identification of the dominant dynamics (\ref{ODE_sys}) of the original PDE system (\ref{PDE_sys})--(\ref{PDE_sys-bound-initial}). 
%
%%instead of the cubic spline interpolation method. For the original PDE system of (\ref{PDE_sys})--(\ref{PDE_sys-bound-initial}), only the dominant dynamics described by the model (\ref{ODE_sys}) will be identified, and a limited number (i.e., $m$ number) of NN models will be resulted. 
%
%\end{remark}

\section{Fault Detection Scheme} \label{sec_FD}

With the results obtained from the above section, a novel FD scheme will be proposed in this section to achieve rapid FD of system (\ref{ODE_sys}). The associated analysis of FD performance will also be provided. 
 
\subsection{FD Estimator Design and Decision Making} 

With the constant RBF NN models $\bar{W}_i^{0\top}S(x_s,u)$ in (\ref{Identify_funcapp-bar}), a bank of FD estimators can be constructed  as follows: 
\bgeq \label{FD_estimator}
\begin{aligned}
	\dot{\bar{x}}_{i}^0 = -b_{i}^0(\bar{x}_{i}^0-x_{s_i})+ \lambda_i x_{s_i} +& \bar{W}_i^{0\top}S(x_s,u), \, \\
%	& i=1,2,\cdots,m,
\end{aligned}
\edeq
where $i=1, \cdots,m$, $\bar{x}_i^0 $ is the estimator state with initial condition $\bar{x}_i^0(0)=x_{s_i}(0)$, $x_{s_i}$ is the $i$-th state of system (\ref{ODE_sys}), $b_{i}^0$ is a positive design constant, $\lambda_i$ is the $i$-th diagonal element of $A_s$ in (\ref{ODE_sys}), and $\bar{W}_i^{0\top}S(x_s,u)$ is used to approximate the function $f_{s_i}(x_s,u)$ in (\ref{ODE_sys}). 
Comparing the FD estimators (\ref{FD_estimator}) with the monitored system (\ref{ODE_sys}), and based on (\ref{Identify_funcapp-bar}), the following residual system (with residual $\tilde{x}_i^0 := \bar{x}_{i}^0-x_{s_i}$) can be derived: 
\bgeq \label{FD_error-sys}
\begin{aligned}
	\dot{\tilde{x}}_i^0 %= &\, -b_{i}^0 \tilde{x}_i^0 + \bar{W}_i^{0\top}S(x_s,u) - f_{s_i}(x_s, u) - \beta(t-t_0) \phi_{s_i}^{k}(x_s, u) \\
	= &\,  -b_{i}^0 \tilde{x}_i^0 - \varepsilon_i^0 - \beta(t-t_0) \phi_{s_i}^{k}(x_s, u), %\quad i=1,\cdots,m
\end{aligned}
\edeq
where $\varepsilon_i^0$ is the approximation error of model $\bar{W}_i^{0\top}S(x_s,u)$ for function $f_{s_i}(x_s, u)$ as defined in (\ref{Identify_funcapp-bar}), and $\phi_{s_i}^{k}(x_s, u)$ is the faulty dynamics occurring in system (\ref{ODE_sys}). %, and $\bar{W}_i^{0\top}S(x_s,u) - f_{s_i}(x_s,u) = - \varepsilon_i^0 $ can be derived according to (\ref{Identify_funcapp-bar}). 
The $\mathcal{L}_1$ norm of residual signal $\tilde{x}_{i}^0$ in (\ref{FD_error-sys}), i.e., $\left\| \tilde{x}_{i}^0(t) \right\|_1 = \frac{1}{T} \int_{t-T}^t \left| \tilde{x}_i^0 (\tau) \right| d \tau $ ($t>T$) with $T$ being a design parameter, will be used for real-time FD decision making. Before proceeding further, a threshold, denoted as $\bar{e}_i^0$, will be further designed to upper bound $\left\|\tilde{x}_{i}^0(t)\right\|_1$ when the monitored system (\ref{ODE_sys}) is operating in normal mode (i.e., for time $t<t_0$). 
To this end, consider the residual system (\ref{FD_error-sys}) for time $t < t_0$, note that $ \tilde{x}_{i}^0(0) =0$ and $ |\varepsilon_i^0 |<\xi_i^*$, for all $i=1,\cdots,m$, the system state $\tilde{x}_{i}^0$ satisfies: 
\bgeq \label{FD_error-dynamics}
\begin{aligned}
	 &\, \left|\tilde{x}_{i}^0(t) \right| = \left|\tilde{x}_{i}^0(0) e^{-b_i^0 t } - \int^{t}_{0} e^{-b_i^0 (t - \tau)} \varepsilon_i^0 d \tau \right| \\
	&\, \leq \int^{t}_{0} e^{-b_i^0 (t - \tau)} \left| \varepsilon_i^0 \right| d \tau 
	< \int^{t}_{0} e^{-b_i^0 (t - \tau)} \xi_i^* d \tau 
%	= \bar{\epsilon}_i e^{-b_i^0 t } + \frac{\xi_i^*}{b_i^0} ( 1- e^{-b_i^0 t} )
	<  \frac{\xi_i^*}{b_i^0}.  
\end{aligned}
\edeq 
It implies that the FD residual signal $\left\| \tilde{x}_{i}^0(t) \right\|_1 < \frac{\xi_i^*}{b_i^0} $ holds under the normal mode for all time $t < t_0$. 
Based on this, the FD threshold $\bar{e}_i^0 $ can be designed as:
\bgeq \label{FD_threshold}
	\bar{e}_i^0 := \frac{1}{b_i^0} ( \xi_i^* + \varrho_i) ,  \quad i=1,\cdots,m,
\edeq
where $\xi_i^*$ is a small constant given in (\ref{Identify_funcapp-bar}), $b_i^0$ is a design constant from (\ref{FD_estimator}), and $\varrho_i \geq 0$ is a small constant added as an auxiliary parameter for preventing possible FD misjudgment. 

With the FD estimators (\ref{FD_estimator}) and FD thresholds (\ref{FD_threshold}), the FD decision making is based on the following principle:
when no fault occurs in the monitored system (\ref{ODE_sys}), the residuals $\left\|\tilde{x}_{i}^0(t) \right\|_1 $ remain smaller than the corresponding thresholds $\bar{e}_i^0 $ for all $ i=1,\cdots,m$. 
If there exists a time instant $t_d$, such that, for some $i \in \{1,\cdots,m\}$, the FD residuals $\left\|\tilde{x}_{i}^0(t_d) \right\|_1 $ become larger than the corresponding thresholds $\bar{e}_i^0 $, i.e., $\left\|\tilde{x}_{i}^0(t_d) \right\|_1 > \bar{e}_i^0 $, it indicates that a certain fault must occur in the system (\ref{ODE_sys}). 
As a result, the occurrence of fault can be detected at time $t_d$. The idea is formalized as follows: 

\textbf{Fault detection decision making}: Compare the FD residual signals $\left\|\tilde{x}_i^0(t) \right\|_1$ with the FD thresholds $\bar{e}_i^0 $ for all $i=1,\cdots,m$. If there exists a finite time $t_d$, such that, for some $i \in \{1,\cdots,m\}$, $\left\| \tilde{x}_i^0(t_{d}) \right\|_1 > \bar{e}_i^0 $ holds. Then, the occurrence of a fault is deduced at time $t_d$. 

\begin{remark} \label{rem_epsilon*}

The parameter $\xi_i^*$ in (\ref{FD_threshold}) represents the upper bound of steady absolute approximation error $\left| \bar{W}_i^{k\top}S(x_s,u)- \eta_i^k(x_s,u) \right|$ of (\ref{Identify_funcapp-bar}) for all $ k=0,1,\cdots,m$. Direct derivation of this parameter is quite difficult since the function $ \eta_i^k(x_s,u) $ is not available. 
Alternatively, the value of $\xi_i^*$ could be evaluated in the following way: 
in the training phase of Section \ref{sec_identification}, with the obtained constant models $\bar{W}_i^{k\top}S(x_s,u)$, a bank of estimators in the form of (\ref{FD_estimator}) (with $ k = 0,1,\cdots,N $) can be developed by setting $b_i^k=1$. %These estimators are used to approximate the system (\ref{ODE_sys}) under each $k$-th faulty mode. 
Then, following a similar line of the analysis in Eqs. (\ref{FD_estimator})--(\ref{FD_error-dynamics}), it can be proved that the associated state error $\tilde{x}_i^k = \bar{x}_i^k -x_{s_i}$ satisfies: $|\tilde{x}_i^k(t)| < \xi_i^* $. 
Thus, with such estimators, the value of $\xi_i^*$ can be obtained as the upper bound of steady absolute state error $|\tilde{x}_i^k|$, $\forall k = 0,1,\cdots,N $.  

\end{remark}

\begin{remark} \label{rem_varrho}
The parameter $\varrho_i$ in the FD threshold (\ref{FD_threshold}) is designed to improve robustness against system uncertainties.  
%
%Note that during the real-time operation, the monitored system (\ref{ODE_sys}) could suffer disturbances resulting in deviation of system dynamics. 
%
More specific, note that although our FD scheme is developed based on the approximate ODE system (\ref{ODE_sys}), the real-time FD process will be carried out on the original PDE system (\ref{ODE_sys-sf0}). 
The associated system dynamics of models (\ref{ODE_sys}) and (\ref{ODE_sys-sf0}) have a small difference due to the fast dynamics of state $\chi_f$. 
%There would be small deviation of system dynamics $x_s$ between models (\ref{ODE_sys}) and (\ref{ODE_sys-sf0}) due to the effect of (neglected) fast dynamics $\chi_f$. 
%
The parameter $\varrho_i$ is thus introduced to compensate such a difference and mitigate its potential effects on the FD performance. 
%This could render the FD residual $\left\| \tilde{x}_i^0 \right\|_1$ of (\ref{FD_error-dynamics}) increasing and exceeding the threshold $\frac{\xi_i^*}{b_i^0}$ even without the presence of faults, leading to a possible false alarm of FD. 
%%
%Adding an auxiliary parameter $\varrho_i$ in the FD threshold (\ref{FD_threshold}) is to settle this issue. %, to guarantee desired FD reliability. 
%%
%%However, for such a parameter $\varrho_i$, setting a large value could result in failure of achieving detection for the faults with relatively small magnitude, leading to limited fault detectability. 
%%%
%%Thus, when determining the value of $\varrho_i$, the tradeoff between FD reliability and detectability should be taken into account. 

\end{remark}

\begin{remark} \label{rem_FDthreshold}
The FD threshold (\ref{FD_threshold}) can be made very small, because the parameter $\xi_i^*$ can be made arbitrarily small by constructing a sufficiently large number of neurons in the training process of Section \ref{sec_identification}, and the parameter $\varrho_i$ can be selected also as a very small number. 

\end{remark}

\begin{remark} \label{rem_FDestimator}

Most of existing FD schemes (e.g., \cite{CaiFS.AUT16, DeyPM.AUT19, ArmD.AIChE08, DemA.JRNC12}) cannot deal with the effect of system uncertainty on FD process. As a result, these schemes require the occurring faults to be of sufficiently large magnitudes (larger than that of system uncertainty) for successful detection, quite limiting their fault detectability. 
However, this issue can be addressed under our scheme. Specifically, as established in Section \ref{sec_identification}, the system uncertainty $f_{s_i}(x_s,u)$ in (\ref{ODE_sys}) can be accurately identified with the DL-based dynamics learning scheme and the associated knowledge can be obtained and stored in a constant NN model $\bar{W}_i^{0\top}S(x_s,u)$ of (\ref{Identify_funcapp-bar}). By using this model to design the FD estimator (\ref{FD_estimator}), the fault dynamics $\phi_{s_i}^k(x_s,u)$ in system (\ref{ODE_sys}) can be accurately distinguished from the uncertain dynamics $f_{s_i}(x_s,u)$, and will be captured by the FD signal $\tilde{x}_i^0$ in (\ref{FD_error-sys}) for accurate detection. It will facilitate our scheme to develop improved fault detectability compared to the ones in \cite{CaiFS.AUT16, DeyPM.AUT19, ArmD.AIChE08, DemA.JRNC12}. Associated rigorous analysis will be conducted in the next section.

\end{remark}

\begin{remark} \label{rem_FDfast}

Rapid FD process can be achieved with our approach. Note that the FD estimators of (\ref{FD_estimator}) are designed with the constant NN models $\bar{W}_i^{0\top}S(x_s,u)$, whose implementation does not involve any parameter adaptation. This will largely shorten the FD time, such that FD process of (\ref{FD_estimator}) can be achieved in a rapid manner. 

\end{remark}

\subsection{Detectability Condition}

To analyze the performance of the proposed FD scheme, in the following, we will study the fault detectability condition, i.e., under what conditions the occurring fault in system (\ref{ODE_sys}) is detectable with our proposed FD scheme. 

\begin{theorem} \label{theo_FD}
Consider the system (\ref{ODE_sys}) and the fault detection system consisting of estimators (\ref{FD_estimator}) and thresholds (\ref{FD_threshold}). If there exists a time interval $ I = [t_{a}, t_{b} ] \subseteq [t_b-T,t_b]$ with $t_a \geq t_0$, such that for some $i \in \{1,\cdots,m\}$, 
\bgeq \label{FD_detcond1}
	\left| \phi_{s_i}^k (x_s (t), u (t)) \right| > 2\xi_i^* + 2 \varrho_i, \quad \forall t \in I
\edeq
and 
\bgeq \label{FD_detcond2}
	l: = t_{b}-t_{a} \geq \frac{1}{b_i^0} \textup{ln} \frac{7 \mu_i - 6\xi_i^*}{\mu_i - 2\xi_i^*} + \frac{T( 4 \xi_i^* + 4 \varrho_i) }{ 3 \mu_i - 2\xi_i^*},
\edeq
%with $b_i^0 \geq \frac{3 \mu_i +\xi_i^*}{T(\mu_i-\xi_i^*)} \textup{ln} \frac{7 \mu_i + \xi_i^*}{\mu_i - \xi_i^*}$, 
where $\mu_i := \min \{\left| \phi_{s_i}^k (x_s,u) \right|, \forall t \in I \}  $, 
then, $ \left\|\tilde{x}_i^0(t_b) \right\|_1 > \bar{e}_i^0 $ holds and the occurrence of a fault will be detected at time $t_b$, i.e., $t_d=t_b$. 
\end{theorem} 

\begin{proof}
Consider the residual signal $\tilde{x}_{i}^0(t) $ of (\ref{FD_error-sys}). 
In the time interval $I$, we assume that there exists a subinterval $I' \subseteq I$ such that the signal $\left| \tilde{x}_{i}^0(t) \right|$ has a very small magnitude, i.e.,
\bgeq \label{FD_dettheo1}
	I': = \left\lbrace t \in I: \, \left| \tilde{x}_{i}^0(t) \right| \leq \frac{3 \mu_i - 2 \xi_i^*}{4 b_i^0} \right\rbrace, 
\edeq
where $\mu_i > 2 \xi_i^*+2 \varrho_i  $ from (\ref{FD_detcond1}). 
For $t \in I'$, by denoting $t_{a}'= \min\{t, t \in I'\}$, the residual signal  $\tilde{x}_{i}^0 $ of (\ref{FD_error-sys}) satisfies  
\bgeq \label{FD_dettheo2}
\begin{aligned}
	%&\, \left| \tilde{x}_{i}^0(t) \right| \\ %= &\, \tilde{x}_{i}^0(0) e^{-b_i^0 t } + \int^{t}_{0} e^{-b_i^0 (t - \tau)} (\bar{W}_i^{0\top}S(x_s,u) - f_{s_i}(x_s,u)) d \tau - \int^{t}_{t_0} e^{-b_i^0 (t - \tau)} \phi_{s_i}^{l'}(x_s,u) d \tau \\
	\left| \tilde{x}_{i}^0 \right| & = \left| \tilde{x}_{i}^0(t_{a}') e^{-b_i^0 (t-t_{a}') } - \int^{t}_{t_{a}'} e^{-b_i^0 (t - \tau)} ( \phi_{s_i}^{k}(x_s, u) + \varepsilon_i^0 ) d \tau \right| \\
	\geq & \left| \int^{t}_{t_{a}'} e^{-b_i^0 (t - \tau)} ( \phi_{s_i}^{k}(x_s, u) + \varepsilon_i^0 ) d \tau \right| - \left| \tilde{x}_{i}^0(t_{a}') \right| e^{-b_i^0 (t-t_{a}') } . \\
%	\geq &\, \left| \int^{t}_{t_{a}'} e^{-b_i^0 (t - \tau)} ( \phi_{s_i}^{k}(x_s, u) + \varepsilon_i^0 ) d \tau \right| - \frac{3 \mu_i - 2 \xi_i^*}{4 b_i^0} e^{-b_i^0 (t-t_{a}') }. 
\end{aligned}
\edeq 
From (\ref{FD_detcond1}) and (\ref{Identify_funcapp-bar}), for all $t \in I'$, $ \phi_{s_i}^{k}(x_s, u) + \varepsilon_i^0$ satisfies
\bgeq \label{FD_dettheo3}
\begin{aligned}
	\left| \phi_{s_i}^{k}(x_s, u) + \varepsilon_i^0 \right| &\,\geq \left| \phi_{s_i}^{k}(x_s, u) \right| - \left| \varepsilon_i^0 \right| \geq \mu_i -\xi_i^*.  
\end{aligned}
\edeq
Note that $\mu_i-\xi_i^*>0$, it is easily seen that $ \phi_{s_i}^{k}(x_s, u) + \varepsilon_i^0$ has an unchanged sign for all $t \in I'$, such that %$\left| \int^{t}_{t_{a}'} e^{-b_i^0 (t - \tau)} ( \phi_{s_i}^{k}(x_s, u) + \varepsilon_i^0 ) d \tau \right| =  \int^{t}_{t_{a}'} e^{-b_i^0 (t - \tau)} \left| ( \phi_{s_i}^{k}(x_s, u) + \varepsilon_i^0 ) \right| d \tau$ holds. 
\bgeq \label{FD_dettheo30}
\begin{aligned}
	&\,\left| \int^{t}_{t_{a}'} e^{-b_i^0 (t - \tau)} ( \phi_{s_i}^{k}(x_s, u) + \varepsilon_i^0 ) d \tau \right| \\ 
	= &\, \int^{t}_{t_{a}'} e^{-b_i^0 (t - \tau)} \left| \phi_{s_i}^{k}(x_s, u) + \varepsilon_i^0 \right| d \tau \\
	 \geq &\, \int^{t}_{t_{a}'} e^{-b_i^0 (t - \tau)} (\mu_i -\xi_i^*) d \tau 
	 = \frac{\mu_i - \xi_i^*}{b_i^0} (1 - e^{-b_i^0(t-t_{a}')} ). 
\end{aligned}
\edeq 
Then, since $\left| \tilde{x}_{i}^0(t_{a}') \right| \leq \frac{3 \mu_i - 2 \xi_i^*}{4 b_i^0}$, inequality (\ref{FD_dettheo2}) reduces to 
\bgeq \label{FD_dettheo4}
\begin{aligned}
	\left| \tilde{x}_{i}^0(t) \right| 
	%\geq &\,  \int^{t}_{t_{a}'} e^{-b_i^0 (t - \tau)} \left| \phi_{s_i}^{k}(x_s, u) + \varepsilon_i^0 \right| d \tau - \frac{3 \mu_i - 2 \xi_i^*}{4 b_i^0} e^{-b_i^0 (t-t_{a}') } \\
	%\geq &\, \int^{t}_{t_{a}'} e^{-b_i^0 (t - \tau)} (\mu_i -\xi_i^*) d \tau - \frac{3 \mu_i - 2 \xi_i^*}{4 b_i^0} e^{-b_i^0 (t-t_{a}') } \\
	\geq &\, \frac{\mu_i - \xi_i^*}{b_i^0} (1 - e^{-b_i^0(t-t_{a}')} ) - \frac{3 \mu_i - 2 \xi_i^*}{4 b_i^0} e^{-b_i^0( t - t_{a}')} . 
\end{aligned}
\edeq
As a result, it can be deduced that $\left| \tilde{x}_{i}^0(t) \right| > \frac{3 \mu_i - 2 \xi_i^*}{4 b_i^0}$ holds for $t - t_{a}' > \frac{1}{b_i^0} \textup{ln} \frac{7 \mu_i - 6\xi_i^*}{\mu_i - 2\xi_i^*}$, and the length of time interval $I'$ in (\ref{FD_dettheo1}), denoted by $l'$, satisfies $ l' < \frac{1}{b_i^0} \textup{ln} \frac{7 \mu_i - 6\xi_i^*}{\mu_i - 2\xi_i^*}$. 
Furthermore, it is easily verified that there exists at most one subinterval $I'$ defined in (\ref{FD_dettheo1}) over the time interval $I$. This implies that for the time interval $I-I'$, we have:
\bgeq \label{FD_dettheo6}
	\left| \tilde{x}_i^0 (t) \right| > \frac{3 \mu_i - 2 \xi_i^*}{4 b_i^0}, \quad \forall t \in I-I',
\edeq
and the length of time interval $I-I'$ (i.e., $l-l'$) satisfies $l - l' > l- \frac{1}{b_i^0} \textup{ln} \frac{7 \mu_i - 6\xi_i^*}{\mu_i - 2\xi_i^*}$. 
%\bgeq \label{FD_dettheo60}
%	l - l' > l- \frac{1}{b_i^0} \textup{ln} \frac{7 \mu_i - 6\xi_i^*}{\mu_i - 2\xi_i^*}.
%\edeq
Based on this, from (\ref{FD_detcond2}), we have: 
\bgeq \label{FD_dettheo7}
\begin{aligned}
	\left\| \tilde{x}_{i}^0(t_b) \right\|_1 &\, = \frac{1}{T} \int_{t_b-T}^{t_b} \left| \tilde{x}_{i}^0(\tau) \right| d \tau  
	\geq \frac{1}{T} \int_{I-I'} \left| \tilde{x}_{i}^0(\tau) \right| d \tau \\
	> &\, \frac{1}{T} \int_{I-I'} \frac{3 \mu_i - 2 \xi_i^*}{4 b_i^0} d \tau 
	= \frac{1}{T} ( l - l' ) \frac{3 \mu_i - 2 \xi_i^*}{4 b_i^0} \\
	> &\, \frac{1}{T} ( l - \frac{1}{b_i^0} \textup{ln} \frac{7 \mu_i - 6\xi_i^*}{\mu_i - 2\xi_i^*} ) \frac{3 \mu_i - 2 \xi_i^*}{4 b_i^0}  
	\geq \frac{ \xi_i^*+\varrho_i }{ b_i^0} .
\end{aligned}
\edeq
Thus, $\left\| \tilde{x}_{i}^0(t_b) \right\|_1 > \bar{e}_i^0$ holds and the occurrence of fault in system (\ref{ODE_sys}) can be detected at time $t_b$. This ends the proof. 
\end{proof}

%\begin{remark} \label{rem_FD}
%
%Note that in (\ref{FD_detcond1}), the parameter $\xi_i^*$ can be made arbitrarily small by constructing a sufficiently large number of neurons in the training process of Section \ref{sec_identification}, and the parameter $\varrho_i$ is a very small number that can be freely selected. 
%%
%Thus, conditions (\ref{FD_detcond1})--(\ref{FD_detcond2}) imply that our FD scheme is capable of detecting faults with small magnitudes, which significantly advances existing FD methods of \cite{ArmD.AIChE08, DemA.JRNC12}. %, resulting in satisfactory fault detectability. 
%
%\end{remark}

\begin{remark} \label{rem_FD}
%Theorem \ref{theo_FD} characterizes the class of faults that can be detected under our FD scheme. 
The detectability conditions (\ref{FD_detcond1})--(\ref{FD_detcond2}) show that 
%for successful detection, the fault $\phi_{s_i}^k(x_s,u)$ is required to have sufficiently large magnitude (larger the lower bound $2\xi_i^*+2\varrho_i$) for a sufficiently long time time interval $[t_a,t_b]$ (satisfying (\ref{FD_detcond2})). 
if there exists a time interval $[t_a,t_b]$ of (\ref{FD_detcond2}) such that the occurring fault $\phi_{s_i}^k(x_s,u)$ has a sufficiently large magnitude, i.e., larger than the lower bound $2\xi_i^*+2\varrho_i$ of (\ref{FD_detcond1}), then, fault detection can be achieved. 
Particularly, note that the lower bound $2\xi_i^*+2\varrho_i$ can be made arbitrarily small, as argued in Remark \ref{rem_FDthreshold}, the conditions (\ref{FD_detcond1})--(\ref{FD_detcond2}) are thus satisfiable even for those faults with relatively small magnitudes. %This demonstrates that desired fault detectability can be guaranteed under our FD scheme, which justifies the discussion in Remark \ref{rem_FDestimator}. 
%
%Furthermore, the conditions (\ref{FD_detcond1})--(\ref{FD_detcond2}) also demonstrate that the system uncertainty $f_{s_i}(x_s,u)$ in (\ref{ODE_sys}) has a negligible effect on FD performance, implying that its effect can be adequately dealt with under our FD scheme. 
%%
%All these analyses are consistent with the discussion in Remark \ref{rem_FDestimator}. 

\end{remark}

\section{Fault Isolation Scheme} \label{sec_FI}

Once the occurring fault is detected at time $t_d$ ($t_d>t_0$), the FI scheme will be activated to identify the type of the occurring fault. This section will present the design of such an FI scheme, as well as the associated analysis of FI performance. 
To ease the presentation, we assume without loss of generality that an unknown fault $l'$ that is similar to (but not necessarily perfectly match) the trained fault $l$ ($l \in \{1,\cdots,N\}$) is occurring in system (\ref{ODE_sys}), i.e., 
\bgeq \label{FI_ODE-sys}
\begin{aligned}
	\dot{x}_{s} = A_s x_{s} + f_{s}(x_s, u)+\phi^{l'}_{s}(x_s, u).  %\quad t \geq t_d
%	&\dot{x}_{s_i} = \lambda_i x_{s_i} + f_{s_i}(x_s,u) + \beta(t-t_0) \phi_{s_i}^{l'}(x_s,u), \quad i=1,\cdots,m.  
\end{aligned}
\edeq 

\subsection{FI Estimator Design and Decision Making} \label{sec_FIestimator}

%Consider the system (\ref{FI_ODE-sys}). 
%
With the constant models $\bar{W}_i^{k\top}S(x_s,u)$ of (\ref{Identify_funcapp-bar}) obtained from the identification phase of Section \ref{sec_identification}, we propose to construct a bank of FI estimators in the following form
\bgeq \label{FI_estimator}
\begin{aligned}
	\dot{\bar{x}}_{i}^k = -b_{i}(\bar{x}_{i}^k-x_{s_i}) &\, + \lambda_i x_{s_i}+ \bar{W}_i^{k\top}S(x_s,u),  \\
\end{aligned}
\edeq 
where $i=1,\cdots,m,\, k=1,\cdots,N$, $\bar{x}_i^k $ is the estimator state with initial condition $\bar{x}_i^k(t_d) = x_{s_i}(t_d)$, $x_{s_i}$ is the $i$-th state of system (\ref{FI_ODE-sys}), $\lambda_i$ is the $i$-th diagonal element of matrix $A_s$ in (\ref{FI_ODE-sys}), $b_{i}$ is a positive design constant, and $\bar{W}_i^{k\top}S(x_s,u)$ approximates the function $\eta_i^k(x_s,u)= f_{s_i}(x_s, u) + \phi_{s_i}^{k}(x_s, u)$ of system (\ref{Identify_sys0}). 
Comparing the FI estimators (\ref{FI_estimator})  with the monitored system (\ref{FI_ODE-sys}), and based on (\ref{Identify_funcapp-bar}), the residual systems (with residual $\tilde{x}_i^k := \bar{x}_{i}^k-x_{s_i}$) can be derived as follows: 
\bgeq \label{FI_error-sys}
\begin{aligned}
	\dot{\tilde{x}}_i^k %=& \, -b_{i} \tilde{x}_i^k + \bar{W}_i^{k\top}S(x_s,u) - f_{s_i}(x_s,u) - \phi_{s_i}^{l'}(x_s,u) \\
%	=& \, -b_{i} \tilde{x}_i +  \eta_i^k(x_s,u) - \varepsilon_{i}^k - f_{s_i}(x_s,u) - \phi_{s_i}^k(x_s,u) \\
	=& \, -b_{i} \tilde{x}_i^k - \varepsilon_{i}^k + \phi_{s_i}^{k}(x_s, u)-\phi_{s_i}^{l'}(x_s, u),
\end{aligned}
\edeq
%for all $i = 1,\cdots,m$ and $k=1,\cdots,N$, 
where $\varepsilon_{i}^k$ is the approximation error of model $\bar{W}_i^{k\top}S(x_s,u)$ for function $\eta_i^k(x_s,u)$ as defined in (\ref{Identify_funcapp-bar}), $ \phi_{s_i}^{k}(x_s, u) $ is the $k$-th faulty dynamics that has been learned/trained in Section \ref{sec_identification}, and $\phi_{s_i}^{l'}(x_s, u)$ is the faulty dynamics occurring in system (\ref{FI_ODE-sys}). 
For the purpose of analysis, we introduce a so-called fault mismatch function $ \rho_i^{k,l'}(x_s,u):= \phi_{s_i}^{k}(x_s,u)-\phi_{s_i}^{l'}(x_s,u) $ to represent the dynamics difference between the trained fault $k$ and occurring fault $l'$. 
Then, the residual system (\ref{FI_error-sys}) can be rewritten as:
\bgeq \label{FI_residual}
\begin{aligned}
	\dot{\tilde{x}}_i^k =& \, -b_{i} \tilde{x}_i^k - \varepsilon_{i}^k + \rho_i^{k,l'}(x_s, u). 
\end{aligned}
\edeq 
%for all $i = 1,\cdots,m$ and $k=1,\cdots,N$. 
Similar to the FD case, the $\mathcal{L}_1$ norm of residual signal $\tilde{x}_{i}^k$ in (\ref{FI_residual}), i.e., $\left\|\tilde{x}_{i}^k(t)\right\|_1 = \frac{1}{T} \int_{t-T}^t |\tilde{x}_i^k (\tau)| d \tau $ with $T$ being a design parameter, will be utilized for real-time FI decision making.
%

%\begin{remark} \label{rem_FI-estimator}
%
%Advanced over many existing FI methods (e.g., \cite{GhaE.AUT09, GhaE.CDC07, BanK.ACC12, CaiJ.ACC16}) which are tailored to isolation of faults with special forms (e.g., actuator faults), our proposed FI scheme is applicable to isolation of general faults for DPSs. 
%
%%The design of FI estimators (\ref{FI_estimator}) does not require the isolated faults $\phi_{s_i}^{k}(x_s,u)$ to be known and have any special form. 
%%%
%%With these estimators, our FI scheme can provide a common framework to address the isolation problem of general faults. 
%%%
%%This is advanced over the existing FI methods of \cite{GhaE.AUT09, GhaE.CDC07, BanK.ACC12, CaiJ.ACC16}, which can be used only for some special classes of faults, e.g., actuator faults. % \cite{GhaE.AUT09, GhaE.CDC07, BanK.ACC12} and/or sensor faults \cite{CaiJ.ACC16}. 
%
%\end{remark}
%

%
In the following, for FI decision making, an adaptive threshold, denoted as $\bar{e}_i^k(t)$, will be further designed to upper bound the residual signal $\left\|\tilde{x}_{i}^k(t)\right\|_1$ when the occurring fault $l'$ in (\ref{FI_ODE-sys}) is similar to the trained fault $l$. 
To this end, the following assumption on the dynamics difference between the occurring fault $l'$ and similar fault $l$ is made.  
\begin{assumption} \label{ASS_fault-bound}
The dynamics difference between any pair of the occurring fault $l'$ and its similar  fault $l$ $(l \in \{1,\cdots,N\})$, denoted by $\rho_i^{l,l'}(x_s, u)$, is bounded by a known function $\bar{\rho}_i^l(x_s, u)$, i.e., $\left| \rho_i^{l,l'}(x_s,u) \right| \leq \bar{\rho}_i^l(x_s,u)$ for all $i=1,\cdots,m$. 
\end{assumption}

\begin{remark} \label{rem_comp-rho}

Assumption \ref{ASS_fault-bound} indicates that the occurring fault $l'$ is allowed to have a certain degree of difference from its similar fault $l$, and such difference can be quantified by the function $\bar{\rho}_i^l(x_s,u)$. 
In other words, it allows that the occurring fault is not necessarily required to exactly match any of the pre-defined/pre-trained faults, which however is typically required by existing methods of \cite{CheW.IJACSP14, CheWH.TNNLS13, ZhaY.ACCESS19}. 
%
%It indicates that the occurring fault $l'$ is allowed to have a certain degree of differences from its matched fault $l$, and the allowable fault difference can be quantified by the function $\bar{\rho}_i^l(x_s,u)$. 
%
This property renders our FI scheme a better robust capability of preventing false/missed FI alarm in the presence of slight fault difference during the FI process.  

\end{remark}

Based on the above setup, to design  the FI adaptive threshold, we consider the $l$-th residual system in (\ref{FI_residual}), its time-domain solution can be derived as:
\bgeq \label{FI_match-residual0}
\begin{aligned}
	\tilde{x}_i^l (t) = \tilde{x}_i^l(t_d) e^{-b_i (t-t_d)} + \int_{t_d}^t e^{-b_i(t-\tau)} ( \rho_i^{l,l'}(x_s, u) - \varepsilon_{i}^l ) d \tau .
	%< &\, \left| \tilde{x}_i^l(t_d) \right| e^{-b_i (t-t_d)} + \int_{t_d}^t e^{-b_i(t-\tau)} \left| \varepsilon_{i}^l \right| d \tau.
\end{aligned}
\edeq 
%for all $i=1,\cdots,m$. 
Note that $ \tilde{x}_{i}^l(t_d) =0$ and $\left| \varepsilon_{i}^l \right| < \xi_i^* $ from (\ref{Identify_funcapp-bar}), under Assumption \ref{ASS_fault-bound}, we have: 
\bgeq \label{FI_match-residual}
\begin{aligned}
	\left| \tilde{x}_i^l (t) \right| %= &\, \left| \int_{t_d}^t e^{-b_i(t-\tau)} ( \rho_i^{l,l'}(x_s, u) - \varepsilon_{i}^l ) d \tau \right| \\
%	\leq &\, \int_{t_d}^t e^{-b_i(t-\tau)} \left| \eta_i^l(x_s,u)- \eta_i^l(x_s,\chi_f,u) - \varepsilon_{i}^l + \rho_i^{l,l'}(x_s,\chi_f,u) \right| d \tau \\
	\leq &\, \int_{t_d}^t e^{-b_i(t-\tau)} ( \left| \varepsilon_{i}^l  \right| + \left| \rho_i^{l,l'}(x_s, u) \right| ) d \tau \\
	%< &\, \int_{t_d}^t e^{-b_i(t-\tau)} ( \xi_i^* + \bar{\rho}_i^{l}(x_s, u) ) d \tau \\
	< &\, \frac{\xi_i^*}{b_i} + \int_{t_d}^t e^{-b_i(t-\tau)} \bar{\rho}_i^{l}(x_s, u) d \tau.
\end{aligned}
\edeq 
It guarantees that the FI residual signal $\left\| \tilde{x}_{i}^l(t) \right\|_1$ satisfies: 
\bgeq \label{FI_error-dynamics1}
	\left\| \tilde{x}_{i}^l(t) \right\|_1 < \frac{\xi_i^*}{b_i} + \left\| \int_{t_d}^t e^{-b_i(t-\tau)} \bar{\rho}_i^l (x_s, u) d \tau \right\|_1.
\edeq
%for all $i=1,\cdots,m$. 
Thus, the FI adaptive threshold $\bar{e}_i^l(t)$ can be designed as:
\bgeq \label{FI_threshold}
	\bar{e}_i^l(t) := \frac{\xi_i^*}{b_i} + \left\| \int_{t_d}^t e^{-b_i(t-\tau)} \bar{\rho}_i^{l}(x_s, u) d \tau \right\|_1,
\edeq 
for all $i=1,\cdots,m$, where $\xi_i^*$ is a small constant given in (\ref{Identify_funcapp-bar}), $b_i$ is a design constant from (\ref{FI_estimator}), and $\bar{\rho}_i^{l}(x_s, u)$ is a known function defined in Assumption \ref{ASS_fault-bound}. 

\begin{remark} \label{rem_FIthreshold-simplify}

The FI thresholds (\ref{FI_threshold}) can be implemented in a simplified form if the function $\bar{\rho}_i^l(x_s,u)$ is a constant $\bar{\rho}_i^l$. 
%In a special case of selecting the function $\varrho_i(x_s,\chi_f,u)$ of (\ref{FD_eta}) as a constant $\varrho_i$, 
%
Specifically, with $\bar{\rho}_i^l(x_s,u)=\bar{\rho}_i^l$, the threshold (\ref{FI_threshold}) is given by $\bar{e}_i^l(t) = \frac{\xi_i^*}{b_i} + \left\| \frac{\bar{\rho}_i^l }{b_i} (1-e^{-b_i (t-t_d)}) \right\|_1 $, which can be further simplified as a constant threshold $\bar{e}_i^l = \frac{1}{b_i}(\xi_i^*+ \bar{\rho}_i^l) $. 

\end{remark}

Consequently, for the monitored system (\ref{FI_ODE-sys}), the proposed FI scheme consists of the FI estimators (\ref{FI_estimator}) and the adaptive threshold (\ref{FI_threshold}). 
Real-time FI decision making is based on the following principle.  
If there exists a unique residual system in (\ref{FI_residual}), say the $l$-th one, such that for all $ i=1,\cdots,m$ the residual signals $\tilde{x}_i^l$ satisfy $\left\| \tilde{x}_i^l(t) \right\|_1 \leq \bar{e}_i^l (t) $ for all time $t>t_d$, then, it can be deduced that the occurring fault $l'$ in (\ref{FI_ODE-sys}) is similar to the trained fault $l$. 
Using this idea, the FI decision making scheme can be devised as follows.

\textbf{Fault isolation decision making}: Compare the FI residual signals $\left\|\tilde{x}_i^k (t) \right\|_1$ with the FI adaptive thresholds $\bar{e}_i^k (t)$ for time $t \geq t_d$ and all $ i=1,\cdots,m$, $k=1,\cdots,N$. 
If there exists a unique $l \in \{1,\cdots,N\}$ such that: (i) $\forall i=1,\cdots,m$, $\left\|\tilde{x}_i^l(t)\right\|_1 \leq \bar{e}_i^l(t)$ holds for all time $t \geq t_d$; and (ii) $\forall k \in \{1,\cdots,N\} / \{l\}$, $\exists i\in \{1,\cdots,m\}$, $\left\|\tilde{x}_i^k(t^k)\right\|_1 > \bar{e}_i^k(t^k)$ holds at some time instant $t^k > t_d$. Then, the occurring fault $l'$ can be identified similar to the fault $l$, and the isolation time can be obtained as: $ t_{iso} = \max \left\{ t^k, k \in \{1,\cdots,N\} / \{l\} \right\}$.

\begin{remark} \label{rem_FIestimator}

%Similar to the FD estimators in (\ref{FD_estimator}), the FI estimators of (\ref{FI_estimator}) can effectively deal with the effect of system uncertain dynamics $f_s(x_s,u)$ in (\ref{FI_ODE-sys}) on FI process, such that the associated faulty dynamics $ \rho_i^{k,l'}(x_s, u)$ can be accurately measured and captured by the FI residual $\tilde{x}_i^k$ of (\ref{FI_residual}) for achieving accurate FI. Also, rapid FI process can be achieved with these FI estimators. 
%%
%These advantages are owing to the utilization of constant NN models $\bar{W}_i^{k\top}S(x_s,u)$ (obtained through the training process of Section \ref{sec_identification}) in the design of FI estimators (\ref{FI_estimator}). 
%%
%Detailed explanations are similar to those in Remarks \ref{rem_FDestimator}--\ref{rem_FDfast}. %, thus are omitted here.

Similar to the FD scheme of Section \ref{sec_FD}, the proposed FI scheme can effectively deal with the effect of system uncertainty $f_s(x_s,u)$ in (\ref{FI_ODE-sys}) for accurate isolation, and the associated FI process can be achieved in a rapid manner. 
%
%Improved fault isolatability can be developed with our scheme compared to many existing ones, e.g., \cite{ElFG.AIChE07, CaiJ.ACC16, GhaE.AUT09, GhaE.CDC07}. 
%
This is owing to the utilization of constant NN models $\bar{W}_i^{k\top}S(x_s,u)$ (obtained through the training process of Section \ref{sec_identification}) in the design of FI estimators (\ref{FI_estimator}). 
%
%Detailed explanations are similar to those in Remarks \ref{rem_FDestimator}--\ref{rem_FDfast}. %, thus are omitted here.

\end{remark}

\begin{remark} \label{rem_FIthreshold}

Existing FI schemes in \cite{BanK.ACC12, CaiJ.ACC16, GhaE.AUT09} rely on  ``constant'' thresholds for FI, which would limit the ability to separate the temporal dynamics of different type of faults for accurate isolation. These schemes can be applicable only to the faults that have sufficiently distinct differences. %, resulting in limited fault isolatability. 
%
%For example, the FI scheme in \cite{CaiJ.ACC16} can distinguish the faults of different types, e.g., state fault and actuator fault, but cannot recognize the faults that belong to the same type yet have small difference, e.g., state faults at different locations. 
%
For example, the FI schemes in \cite{BanK.ACC12, GhaE.AUT09} can distinguish the actuator faults occurring at different locations, but cannot recognize the actuator faults that occur at the same location but have different magnitudes. 
Advanced over these schemes, our approach design an ``adaptive'' threshold of (\ref{FI_threshold}) by using the nonlinear function $\bar{\rho}_i^{l}(x_s,u)$ that can accurately specify the similarity of each $l$-th type of faults, as defined in Assumption \ref{ASS_fault-bound}. 
With such a threshold, our FI scheme can achieve accurate isolation even for the faults that have relatively small differences. Improved FI accuracy and fault isolatability with our scheme compared to the ones in \cite{BanK.ACC12, CaiJ.ACC16, GhaE.AUT09} will be demonstrated later. 
%
%Associated rigorous analysis will be conducted in the next section. 

\end{remark}

\begin{remark} \label{rem_FIapplicable}

Our FI scheme provides a unified framework to address the isolation problem of general faults occurring in the PDE system (\ref{PDE_sys})--(\ref{PDE_sys-bound-initial}), which possesses enhanced  applicability compared to most existing FI schemes, e.g., \cite{ElFG.AIChE07, BanK.ACC12, CaiJ.ACC16, GhaE.AUT09, GhaE.CDC07}. 
Specifically, note that the schemes of \cite{ElFG.AIChE07, BanK.ACC12, CaiJ.ACC16, GhaE.AUT09, GhaE.CDC07} are applicable only to some special cases. For example, the FI scheme in \cite{ElFG.AIChE07} is tailored only to systems with precisely known model; the method of \cite{CaiJ.ACC16} is tailored to linear PDE systems; while the approaches in \cite{BanK.ACC12, GhaE.AUT09, GhaE.CDC07} are applicable only for actuator faults. 
Advanced over these approaches, our FI scheme is developed for a class of general PDE system with the form of (\ref{PDE_sys})--(\ref{PDE_sys-bound-initial}), in which the occurring fault $\phi^k(x,u)$ is not required to be of any special type and the system model is allowed to have uncertain nonlinear component $f(x,u)$. 
%Actually, the systems investigated in \cite{ElFG.AIChE07, BanK.ACC12, CaiJ.ACC16, GhaE.AUT09, GhaE.CDC07} can be considered as special cases of the system (\ref{PDE_sys})--(\ref{PDE_sys-bound-initial}). 

%
\end{remark}

\subsection{Isolatability Condition}

To analyze the performance of the proposed FI scheme, the fault isolatability condition will be studied, i.e., under what conditions the occurring fault $l'$ in system (\ref{FI_ODE-sys}) can be identified similar to a unique trained fault $l$. 

\begin{theorem} \label{theo_FI}
Consider the monitored system (\ref{FI_ODE-sys}) and the fault isolation system consisting of estimators (\ref{FI_estimator}) and adaptive thresholds (\ref{FI_threshold}). For each $ k \in \{1,\cdots,N\}/\{l\}$ and some $ i\in \{1,\cdots,m\}$, if there exists a time interval $I^k=[t_a^k,t_b^k] \subseteq [t_b^k-T,t_b^k] $ with $t_a^k \geq t_d$, such that 
\bgeq \label{FI_isocond1}
	\left| \rho_{i}^{k,l'} (x_s,u) \right| > \bar{\rho}_i^{k}(x_s, u) + 2 \xi^*_i , \quad  \forall t \in I^k, 
\edeq  
and 
\bgeq \label{FI_isocond2}
\begin{aligned}
	l^k := t_a^k - t_b^k \geq &\, \frac{  2 \xi_i^* + 2 \bar{\rho}^k_{i_{\max}}}{\mu_i + 2\bar{\rho}^k_{i_{\max}}} ( T + \frac{1}{b_i} \textup{ln} \frac{ \mu_i + 2\bar{\rho}^k_{i_{\max}} + \xi_i^*}{\mu_i - 2\xi_i^*} ) \\
	&\, + \frac{\mu_i - 2 \xi_i^*}{ \mu_i + 2\bar{\rho}^k_{i_{\max}} } ( \frac{1}{b_i} \textup{ln} \frac{ 3\mu_i + 4\bar{\rho}^k_{i_{\max}} }{\mu_i - 2\xi_i^*} ),
\end{aligned}
\edeq
where $ \mu_i := \min \{ \left| \rho_{i}^{k,l'} (x_s,u) \right| - \bar{\rho}_i^{k}(x_s, u), \forall t \in I^k \}$, and $\bar{\rho}^k_{i_{\max}}: = \max \{\bar{\rho}_i^{k}(x_s, u), \forall t \geq t_d \} $,  
then, $ \left\|\tilde{x}_i^k(t_b^k) \right\|_1 > \bar{e}_i^k(t_b^k)$ holds, the occurring fault $l'$ will be identified similar to fault $l$, and the isolation time is obtained as $t_{iso}=\max \left\{t_b^k, \forall k \in \{1,\cdots,N\}/\{l\} \right\} $.
\end{theorem} 

\begin{proof}

Consider the $k$-th residual signal $\tilde{x}_i^k (t)$ of (\ref{FI_residual}) and the associated FI threshold $\bar{e}_i^k(t) $ of (\ref{FI_threshold}), with $ k \in \{1,\cdots,N\}/\{l\}$. 
For the purpose of analysis, we introduce a new variable $\vartheta^k_i(t)$ satisfying 
\bgeq \label{FI_theo-theta1}
\begin{aligned}
	 \vartheta^k_i(t) := \frac{\xi_i^*}{b_i} + \int_{t_d}^t e^{-b_i(t-\tau)} \bar{\rho}_i^{k}(x_s, u) d \tau,  \quad t \geq t_d ,
\end{aligned}
\edeq
and thus we have
\bgeq \label{FI_theo-theta2}
\begin{aligned}
	 \vartheta^k_i(t) \leq &\, \frac{\xi_i^*}{b_i} + \int_{t_d}^t e^{-b_i(t-\tau)} \bar{\rho}^k_{i_{\max}} d \tau 
	  %= &\, \frac{\xi_i^*}{b_i} +  \frac{\bar{\rho}^k_{i_{\max}}}{b_i} (1 - e^{-b_i(t-t_d)} ) 
	  < \frac{\xi_i^*+ \bar{\rho}^k_{i_{\max}}}{b_i} , 
\end{aligned}
\edeq
where $\bar{\rho}^k_{i_{\max}} = \max \{\bar{\rho}_i^{k}(x_s, u), \forall t \geq t_d \} $. 
Obviously, to guarantee that $\left\| \tilde{x}_i^k (t) \right\|_1 > \bar{e}_i^k(t)$ holds at a time $t=t_b^k$, in light of definitions (\ref{FI_threshold}) and (\ref{FI_theo-theta1}), it is necessary to examine the magnitude of $\left| \tilde{x}_{i}^k(t) \right| - \vartheta^k_i(t)$ for  $t \in [t_b^k-T,t_b^k]$. 

We first consider the time interval $ I^k \subseteq [t_b^k-T,t_b^k] $. Assume that there exists a subinterval $I_1^{k} \subseteq I^k$ such that the residual signal $\left| \tilde{x}_{i}^k(t) \right| - \vartheta^k_i(t)$ has a very small magnitude, i.e., 
\bgeq \label{FI_theo-I1}
	I_1^k: = \left\lbrace t \in I^k: \, \left| \tilde{x}_{i}^k(t) \right| - \vartheta^k_i(t) \leq \frac{\mu_i - 2 \xi_i^*}{2 b_i} \right\rbrace,  
\edeq
where $\mu_i = \min \{ \left| \rho_{i}^{k,l'} (x_s,u) \right| - \bar{\rho}_i^{k}(x_s, u), \forall t \in I^k \} > 2 \xi_i^*$ from (\ref{FI_isocond1}). 
For $t \in I_1^k$, by denoting $t_{a}^{k'}= \min\{t, t \in I^k_1\}$, based on  (\ref{Identify_funcapp-bar}), the residual signal $\tilde{x}_{i}^k(t) $ of (\ref{FI_residual}) satisfies 
\bgeq \label{FI_theo-residual1}
\begin{aligned}
	&\, \left| \tilde{x}_i^k (t) \right| \\
	= &\, \left| \tilde{x}_i^k(t_a^{k'}) e^{-b_i (t-t_a^{k'})} + \int_{t_a^{k'}}^t e^{-b_i(t-\tau)} ( \rho_i^{k,l'}(x_s,u) - \varepsilon_{i}^k ) d \tau \right| \\ 
	\geq &\, \left| \int_{t_a^{k'}}^t e^{-b_i(t-\tau)} \rho_i^{k,l'}(x_s,u) d \tau \right| -  \int_{t_a^{k'}}^t e^{-b_i(t-\tau)} \left| \varepsilon_{i}^k  \right| d \tau \\
	&\, - \left| \tilde{x}_i^k(t_a^{k'}) \right| e^{-b_i (t-t_a^{k'})}  \\ 
	> &\, \left| \int_{t_a^{k'}}^t e^{-b_i(t-\tau)} \rho_i^{k,l'}(x_s,u) d \tau \right| - \frac{\xi_i^*}{b_i} - \left| \tilde{x}_i^k(t_a^{k'}) \right| e^{-b_i (t-t_a^{k'})}.   
\end{aligned} 
\edeq
%where $\left| \varepsilon_{i}^k \right| < \xi_i^*$ from (\ref{Identify_funcapp-bar}). 
Under condition (\ref{FI_isocond1}), it is seen that for all $ t \in I_1^k$, $\rho_i^{k,l'}(x_s,u)$ has an unchanged sign. Then, inequality (\ref{FI_theo-residual1}) yields:
\bgeq \label{FI_theo-residual2}
\begin{aligned}
	\left| \tilde{x}_i^k (t) \right| 
	 >&\, \int_{t_a^{k'}}^t e^{-b_i(t-\tau)}  \left| \rho_i^{k,l'}(x_s,u) \right| d \tau - \frac{\xi_i^*}{b_i} \\
	 &\,- \left| \tilde{x}_i^k(t_a^{k'}) \right| e^{-b_i (t-t_a^{k'})}. 
\end{aligned} 
\edeq
From (\ref{FI_theo-theta1}) and (\ref{FI_theo-residual2}), we have:  
\bgeq \label{FI_theo-error1}
\begin{aligned}
	%& \left| \tilde{x}_i^k (t) \right| - \vartheta^k_i(t) \\
	\left| \tilde{x}_i^k (t) \right| & - \vartheta^k_i(t) > \int_{t_a^{k'}}^t e^{-b_i(t-\tau)}  \left| \rho_i^{k,l'}(x_s,u) \right| d \tau - \frac{2 \xi_i^*}{b_i} \\
	& - \int_{t_d}^t e^{-b_i(t-\tau)} \bar{\rho}_i^{k}(x_s, u) d \tau - \left| \tilde{x}_i^k(t_a^{k'}) \right| e^{-b_i (t-t_a^{k'})} \\
%	= & \int_{t_a^{k'}}^t e^{-b_i(t-\tau)} ( \left| \rho_i^{k,l'}(x_s,u) \right| - \bar{\rho}_i^{k}(x_s, u) ) d \tau - \frac{2 \xi_i^*}{b_i} \\
%	& - \int_{t_d}^{t_a^{k'}} e^{-b_i(t-\tau)} \bar{\rho}_i^{k}(x_s, u) d \tau - \left| \tilde{x}_i^k(t_a^{k'}) \right| e^{-b_i (t-t_a^{k'})} \\
	= & \int_{t_a^{k'}}^t e^{-b_i(t-\tau)} ( \left| \rho_i^{k,l'}(x_s,u) \right| - \bar{\rho}_i^{k}(x_s, u) ) d \tau - \frac{2 \xi_i^*}{b_i} \\
	& - ( \vartheta^k_i(t_a^{k'}) - \frac{\xi_i^*}{b_i} ) e^{-b_i (t-t_a^{k'})} - \left| \tilde{x}_i^k(t_a^{k'}) \right| e^{-b_i (t-t_a^{k'})} \\
	\geq & \int_{t_a^{k'}}^t e^{-b_i(t-\tau)} \mu_i d \tau - \frac{ 2\xi_i^*}{b_i}\\
	& - ( 2\vartheta^k_i(t_a^{k'}) + \frac{\mu_i - 4 \xi_i^*}{ 2 b_i} ) e^{-b_i (t-t_a^{k'})} \\
	> & \frac{\mu_i}{b_i} ( 1 - e^{-b_i (t-t_a^{k'})} ) - \frac{ 2\xi_i^*}{b_i} - \frac{4 \bar{\rho}^k_{i_{\max}} + \mu_i}{ 2 b_i} e^{-b_i (t-t_a^{k'})} , 
\end{aligned} 
\edeq
where $\left| \tilde{x}_i^k(t_a^{k'}) \right| \leq \vartheta^k_i(t_a^{k'}) + \frac{\mu_i - 2\xi_i^*}{2 b_i}$, $ \left| \rho_i^{k,l'}(x_s,u) \right| - \bar{\rho}_i^{k}(x_s, u) \geq \mu_i $, and $\vartheta^k_i(t_a^{k'}) <  \frac{\xi_i^*+ \bar{\rho}^k_{i_{\max}}}{b_i}$, according to
(\ref{FI_theo-theta2}) and (\ref{FI_theo-I1}). 
Based on (\ref{FI_theo-error1}), it can be deduced that: $\left| \tilde{x}_i^k (t) \right| - \vartheta^k_i(t) > \frac{\mu_i - 2 \xi_i^*}{2 b_i}$ holds for $t - t_a^{k'} > \frac{1}{b_i} \textup{ln} \frac{ 3\mu_i + 4 \bar{\rho}^k_{i_{\max}}}{\mu_i - 2\xi_i^*}$. Thus, the length of time interval $I_1^k$ in (\ref{FI_theo-I1}), denoted by $l_1$, satisfies:
\bgeq \label{FI_theo-I1length}
\begin{aligned}
	l_1 \leq \frac{1}{b_i} \textup{ln} \frac{ 3\mu_i + 4\bar{\rho}^k_{i_{\max}}}{\mu_i - 2\xi_i^*}. 
\end{aligned} 
\edeq
It is easy to verify that there exists at most one subinterval $I_1^k$ of (\ref{FI_theo-I1}) in the time interval $I^k$, which implies  %for $\forall t \in I^k-I_1^k$, we have $\left| \tilde{x}_{i}^k(t) \right| - \vartheta^k_i(t) > \frac{\mu_i - 2 \xi_i^*}{2b_i}$.
\bgeq \label{FI_theo-I-I1}
\begin{aligned}
	\left| \tilde{x}_{i}^k(t) \right| - \vartheta^k_i(t) > \frac{\mu_i - 2 \xi_i^*}{2b_i}, \quad \forall t \in I^k-I_1^k. 
\end{aligned} 
\edeq

Next, following a similar line of the above analysis, we can further deduce that there exists at most one subinterval $I_2^k$ in the time interval $I_1^k$, such that: 
\bgeq \label{FI_theo-I2}
\begin{aligned}
	 &\, \left| \tilde{x}_{i}^k(t) \right| - \vartheta^k_i(t) \leq 0, \quad \forall t \in I_2^k; \\
	 0< &\, \left| \tilde{x}_{i}^k(t) \right| - \vartheta^k_i(t) \leq \frac{\mu_i - 2 \xi_i^*}{2b_i}  , \quad \forall t \in I_1^k-I_2^k,
\end{aligned}
\edeq
and the length of the time interval $I_2^k$, denoted by $l_2$, satisfies: 
\bgeq \label{FI_theo-I2length}
\begin{aligned}
	l_2 \leq \frac{1}{b_i} \textup{ln} \frac{ \mu_i + 2\bar{\rho}^k_{i_{\max}} + \xi_i^*}{\mu_i - 2\xi_i^*}. 
\end{aligned} 
\edeq
Particularly, note that $\vartheta^k_i(t) < \frac{\xi_i^*+ \bar{\rho}^k_{i_{\max}}}{b_i}$ from (\ref{FI_theo-theta2}), inequality (\ref{FI_theo-I2}) yields:
\bgeq \label{FI_theo-I1I2}
\begin{aligned}
	 - \frac{\xi_i^*+ \bar{\rho}^k_{i_{\max}}}{b_i} < & \left| \tilde{x}_{i}^k(t) \right| - \vartheta^k_i(t) \leq 0, \quad \forall t \in I_2^k. \\
%	 0 < & \left| \tilde{x}_{i}^k(t) \right| - \vartheta^k_i(t) \leq \frac{\mu_i - 2 \xi_i^*}{2b_i}  , \, \forall t \in I_1^k-I_2^k.
\end{aligned}
\edeq

Finally, consider the signal $\left| \tilde{x}_i^k (t) \right| - \vartheta^k_i(t)$ in the time interval $[t_b^k-T,t_b^k]$. Dividing $[t_b^k-T,t_b^k]$ into four subintervals, i.e., $[t_b^k-T,t_b^k]= ([t_b^k-T,t_b^k]-I^k) \cup (I^k-I_1^k) \cup (I_1^k-I_2^k) \cup I_2^k$, from (\ref{FI_theo-theta2}) and (\ref{FI_theo-I1length})--(\ref{FI_theo-I1I2}), we obtain: 
\bgeq \label{FI_theo-error2}
\begin{aligned}
	& \int_{[t_b^k-T,t_b^k]-I^k} ( \left| \tilde{x}_i^k (\tau) \right| - \vartheta^k_i(\tau) ) d \tau \\
	& > \int_{[t_b^k-T,t_b^k]-I^k} ( - \frac{\xi_i^*+ \bar{\rho}^k_{i_{\max}}}{b_i} ) d \tau = - ( T-l^k ) \frac{\xi_i^*+ \bar{\rho}^k_{i_{\max}}}{b_i} ; \\
	& \int_{I^k-I_1^k} ( \left| \tilde{x}_i^k (\tau) \right| - \vartheta^k_i(\tau) ) d \tau > \int_{I^k-I_1^k} \frac{\mu_i - 2\xi_i^*}{2 b_i} d \tau\\
	& = (l^k-l_1 ) \frac{\mu_i - 2\xi_i^*}{2 b_i} > ( l^k - \frac{1}{b_i} \textup{ln} \frac{ 3\mu_i + 4\bar{\rho}^k_{i_{\max}}}{\mu_i - 2\xi_i^*} )  \frac{\mu_i - 2\xi_i^*}{2 b_i} ;\\
	& \int_{I_1^k-I_2^k} ( \left| \tilde{x}_i^k (\tau) \right| - \vartheta^k_i(\tau) ) d \tau > ( l_1-l_2 )0 = 0 ; \\
	& \int_{I_2^k} ( \left| \tilde{x}_i^k (\tau) \right| - \vartheta^k_i(\tau) ) d \tau > \int_{I_2^k} ( - \frac{\xi_i^*+ \bar{\rho}^k_{i_{\max}}}{b_i} ) d \tau \\
	& = - l_2 \frac{\xi_i^*+ \bar{\rho}^k_{i_{\max}}}{b_i} > - ( \frac{1}{b_i} \textup{ln} \frac{ \mu_i + 2\bar{\rho}^k_{i_{\max}} + \xi_i^*}{\mu_i - 2\xi_i^*} )  \frac{\xi_i^*+ \bar{\rho}^k_{i_{\max}}}{b_i} . \\
\end{aligned}
\edeq
%where $l^k$ denotes the length of the time interval $I^k$. 
Based on this, with condition (\ref{FI_isocond2}), we have: 
\bgeq \label{FI_theo-FI}
\begin{aligned}
	&\, \left\| \tilde{x}_i^k (t_b^k) \right\|_1 - \bar{e}^k_i(t_b^k) \\
 %	= & \frac{1}{T} \int_{t_b^k-T}^{t_b^k} \left| \tilde{x}_i^k (\tau) \right| d \tau - \frac{1}{T} \int_{t_b^k-T}^{t_b^k}  \vartheta^k_i(\tau) d \tau \\ 
	= & \frac{1}{T} \int_{([t_b^k-T,t_b^k]-I^k) \cup (I^k-I_1^k) \cup (I_1^k-I_2^k) \cup I_2^k} ( \left| \tilde{x}_i^k (\tau) \right| - \vartheta^k_i(\tau) ) d \tau \\
	> & \frac{1}{T} \Big\lbrace ( \frac{ \mu_i + 2 \bar{\rho}^k_{i_{\max}} }{ 2b_i} ) l^k - ( \frac{1}{b_i} \textup{ln} \frac{ 3\mu_i + 4\bar{\rho}^k_{i_{\max}}}{\mu_i - 2\xi_i^*} ) \frac{ \mu_i - 2 \xi_i^*}{ 2 b_i} \\
	&\, - ( T + \frac{1}{b_i} \textup{ln} \frac{ \mu_i + 2\bar{\rho}^k_{i_{\max}} + \xi_i^*}{\mu_i - 2\xi_i^*} ) \frac{\bar{\rho}^k_{i_{\max}} + \xi_i^*}{b_i} \Big\rbrace \geq 0. \\
%	\geq &\, 0. 
\end{aligned}
\edeq  
Thus, $\left\| \tilde{x}_i^k (t_b^k) \right\|_1 > \bar{e}^k_i(t_b^k) $ holds, and the possibility that the fault $l'$ occurring in system (\ref{FI_ODE-sys}) is similar to  the trained fault $k$ (for any $k\in \{1,\cdots,N\}/\{l\}$) can be excluded. Consequently, the fault $l'$ will be identified similar to the trained fault $l$, and the isolation time is obtained as: $t_{iso}=\max \left\{t_b^k, \forall k \in \{1,\cdots,N\}/\{l\} \right\} $. This ends the proof. 
\end{proof}

\begin{remark}
%%Theorem \ref{theo_FI} characterizes the class of faults that can be isolated under our FI scheme. 
The isolatability conditions (\ref{FI_isocond1})--(\ref{FI_isocond2}) show that, for different types of faults $\phi_i^{l'}(x_s,u)$ and $\phi_i^k(x_s,u)$ ($l\neq k$), if their dynamic difference $\rho_i^{k,l'} (x_s,u)$ has sufficiently large magnitudes (larger than the bound function $\bar{\rho}_i^{k}(x_s, u) + 2 \xi^*_i $) over some time interval $[t_a^k,t_b^k]$, then, these two types of faults can be effectively isolated. 
Essentially, these conditions imply that the FI process is achieved by utilizing the known nonlinear function $\bar{\rho}_i^{k}(x_s, u)$ to separate the faulty dynamics of $\phi_i^{l'}(x_s,u)$ and $\phi_i^k(x_s,u)$. %Note that $\bar{\rho}_i^{k}(x_s, u)$ is a precisely-known nonlinear function that can provide accurate and effective dynamics-separation, desired fault accuracy and isolatability can be guaranteed with our FI scheme. 

\end{remark}

\begin{figure*} [htb!]
\centering
	\begin{subfigure}[t]{0.32 \textwidth}
	\includegraphics[width=1 \textwidth]{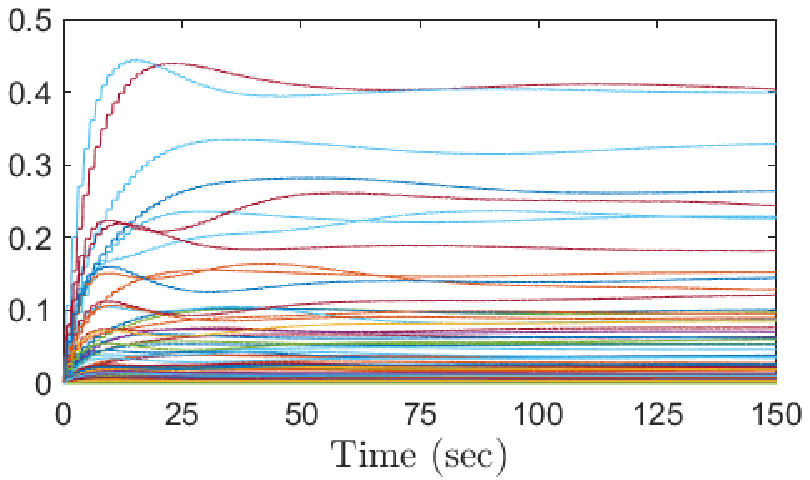}	
	\caption{}
	\label{1a_weight}
	\end{subfigure}
	\begin{subfigure}[t]{0.32 \textwidth}
	\includegraphics[width=1 \textwidth]{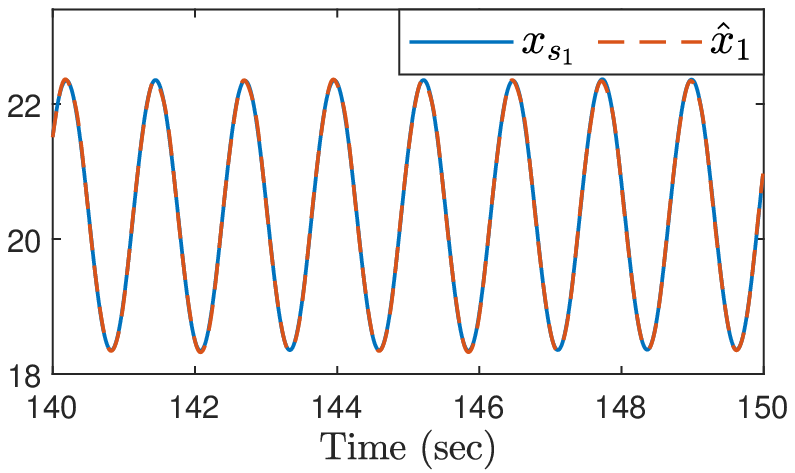}	
	\caption{}
	\label{1b_xhat}
	\end{subfigure}
	\begin{subfigure}[t]{0.32 \textwidth}
	\includegraphics[width=1 \textwidth]{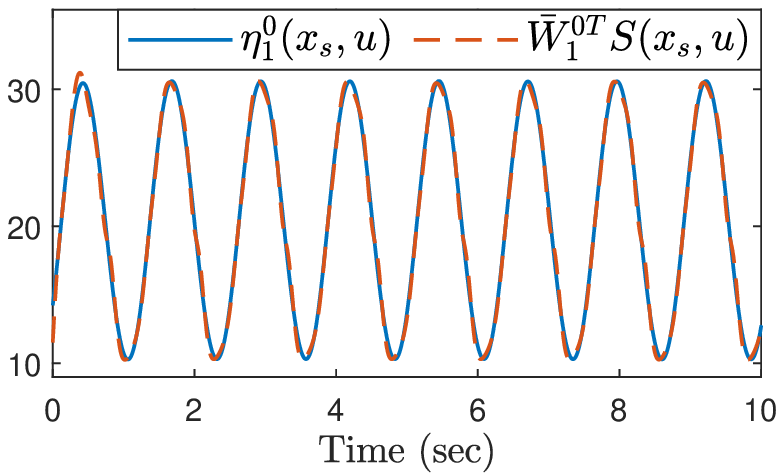}	
	\caption{}
	\label{1d_func}
	\end{subfigure}
	\caption{ Identification of function $\eta_1^0(x_s,u)$ with identifier (\ref{Identify_obs}). (a) Convergence of NN weight $\hat{W}_1^0$; (b) estimation performance of $x_{s_1}$ by $\hat{x}_1$; and (c) function approximation performance of $\eta_{1}^0(x_s,u)$ by $\bar{W}_1^{0\top}S(x_s,u)$.  }
	\label{PLOT_identify-0}
\end{figure*}
%
%\begin{figure*} [htb!]
%\centering
%	\begin{subfigure}[t]{0.3 \textwidth}
%	\includegraphics[width=1.1 \textwidth]{pics/2a_weight.eps}	
%	\caption{}
%	\label{2a_weight}
%	\end{subfigure}
%	\begin{subfigure}[t]{0.3 \textwidth}
%	\includegraphics[width=1.1 \textwidth]{pics/2b_xhat.eps}	
%	\caption{}
%	\label{2b_xhat}
%	\end{subfigure}
%	\begin{subfigure}[t]{0.3 \textwidth}
%	\includegraphics[width=1.1 \textwidth]{pics/2d_func.eps}	
%	\caption{}
%	\label{2d_func}
%	\end{subfigure}
%	\caption{ Identification of function $\eta_1^1(x_s,u)$ with identifier (\ref{Identify_obs}). (a) Convergence of NN weight $\hat{W}_1^1$; (b) estimation performance of $x_{s_1}$ by $\hat{x}_1$; and (c) function approximation performance of $\eta_{1}^1(x_s,u)$ by $\bar{W}_1^{1\top}S(x_s,u)$.  }
%	\label{PLOT_identify-1}
%\end{figure*}

%\begin{figure*} [htb!]
%\centering
%	\includegraphics[width= 0.85 \textwidth]{pics/3_fault1m.eps}	
%	\caption{ FDI performance when fault $1'$ occurs in system (\ref{SIM_sys}). Time profiles of (a) the FD residuals $\left\| \tilde{x}_i^0 (t) \right\|_1$ and FD thresholds $\bar{e}_i^0$; (b) the $1$-st FI residuals $\left\| \tilde{x}_i^1 (t) \right\|_1$ and FI thresholds $\bar{e}_i^1 (t)$; (c) the $2$-nd FI residuals $\left\| \tilde{x}_i^2 (t) \right\|_1$ and FI thresholds $\bar{e}_i^2(t)$; and (d) the $3$-rd FI residuals $\left\| \tilde{x}_i^3 (t) \right\|_1$ and FI thresholds $\bar{e}_i^3 (t)$, $\forall i=1,2,3$. }
%	\label{PLOT_FD-fault1m}
%\end{figure*}
%
\begin{figure*} [htb!]
\centering
	\begin{subfigure}[t]{0.85 \textwidth}
	\includegraphics[width=1 \textwidth]{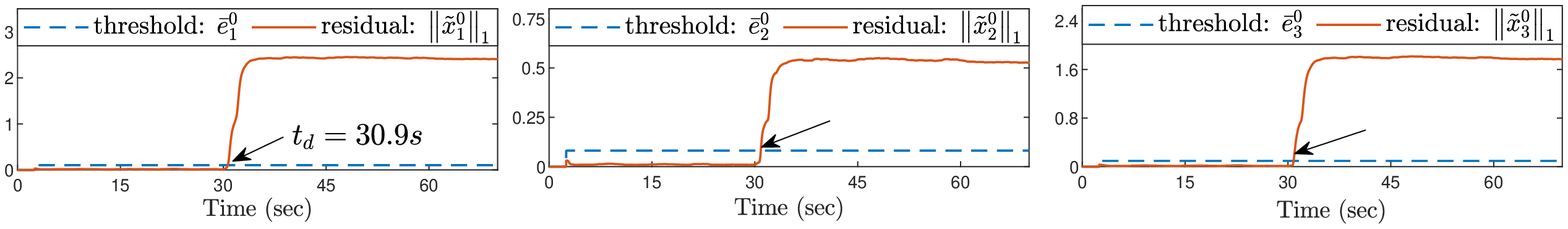}	
	\caption{}
	\label{3_fault1-FD}
	\end{subfigure}
	\begin{subfigure}[t]{0.85 \textwidth}
	\includegraphics[width=1 \textwidth]{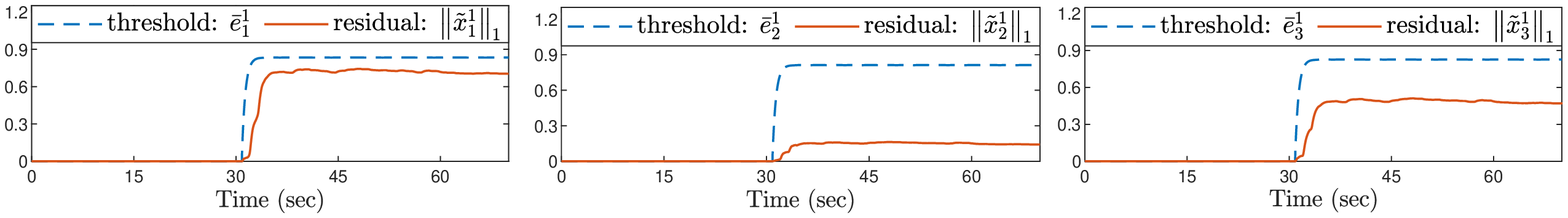}	
	\caption{}
	\label{3_fault1-FI1}
	\end{subfigure}
	\begin{subfigure}[t]{0.85 \textwidth}
	\includegraphics[width=1 \textwidth]{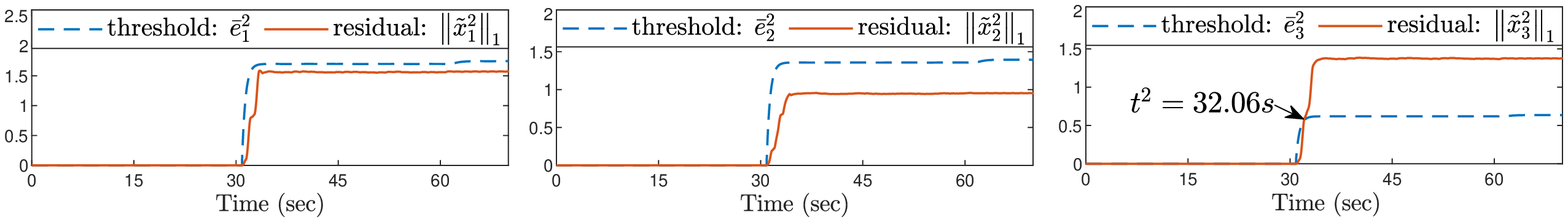}	
	\caption{}
	\label{3_fault1-FI2}
	\end{subfigure}
	\begin{subfigure}[t]{0.85 \textwidth}
	\includegraphics[width=1 \textwidth]{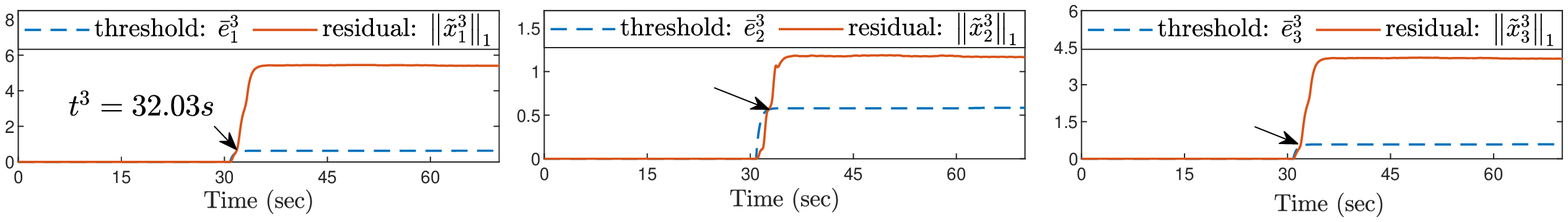}	
	\caption{}
	\label{3_fault1-FI3}
	\end{subfigure}
%	\caption{ FDI performance when fault $1'$ occurs in system (\ref{SIM_sys}). Time profiles of (a) the FD residuals $\left\| \tilde{x}_i^0 (t) \right\|_1$ and FD thresholds $\bar{e}_i^0$; (b) the $1$-st FI residuals $\left\| \tilde{x}_i^1 (t) \right\|_1$ and FI thresholds $\bar{e}_i^1 (t)$; (c) the $2$-nd FI residuals $\left\| \tilde{x}_i^2 (t) \right\|_1$ and FI thresholds $\bar{e}_i^2(t)$; and (d) the $3$-rd FI residuals $\left\| \tilde{x}_i^3 (t) \right\|_1$ and FI thresholds $\bar{e}_i^3 (t)$, $\forall i=1,2,3$. }
	\caption{ FDI performance when fault $1'$ occurs at time $t_0=30s$ in system (\ref{SIM_sys}): (a) FD residuals and thresholds; (b) $1$-st FI residuals and thresholds; (c) $2$-nd FI residuals and thresholds; and (d) $3$-rd FI residuals and thresholds. }
	\label{PLOT_FD-fault1m}
\end{figure*}

%\begin{figure*} [htb!]
%\centering
%	\includegraphics[width= 0.85 \textwidth]{pics/4_fault2m.eps}	
%	\caption{ FDI performance when fault $2'$ occurs in system (\ref{SIM_sys}). Time profiles of (a) the FD residuals $\left\| \tilde{x}_i^0 (t) \right\|_1$ and FD thresholds $\bar{e}_i^0 $; (b) the $1$-st FI residuals $\left\| \tilde{x}_i^1 (t) \right\|_1$ and FI thresholds $\bar{e}_i^1 (t)$; (c) the  $2$-nd FI residuals $\left\| \tilde{x}_i^2 (t) \right\|_1$ and FI thresholds $\bar{e}_i^2 (t)$; and (d) the $3$-rd FI residuals $\left\| \tilde{x}_i^3 (t) \right\|_1$ and FI thresholds $\bar{e}_i^3 (t)$, $\forall i=1,2,3$. }
%	\label{PLOT_FD-fault2m}
%\end{figure*}

\begin{figure*} [htb!]
\centering
	\begin{subfigure}[t]{0.85 \textwidth}
	\includegraphics[width=1 \textwidth]{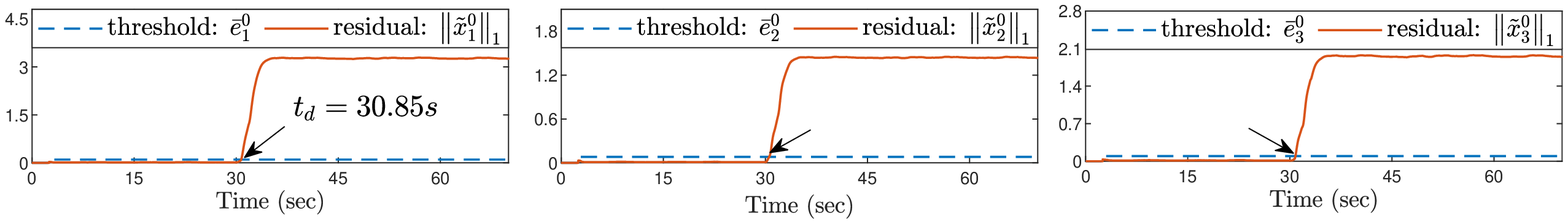}	
	\caption{}
	\label{4_fault2-FD}
	\end{subfigure}
	\begin{subfigure}[t]{0.85 \textwidth}
	\includegraphics[width=1 \textwidth]{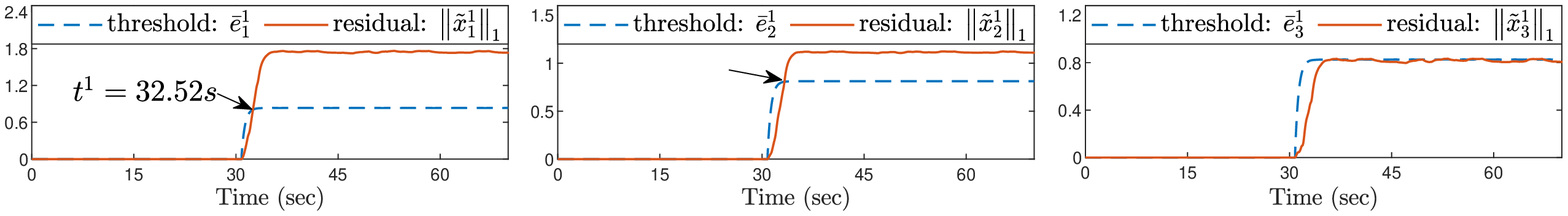}	
	\caption{}
	\label{4_fault2-FI1}
	\end{subfigure}
	\begin{subfigure}[t]{0.85 \textwidth}
	\includegraphics[width=1 \textwidth]{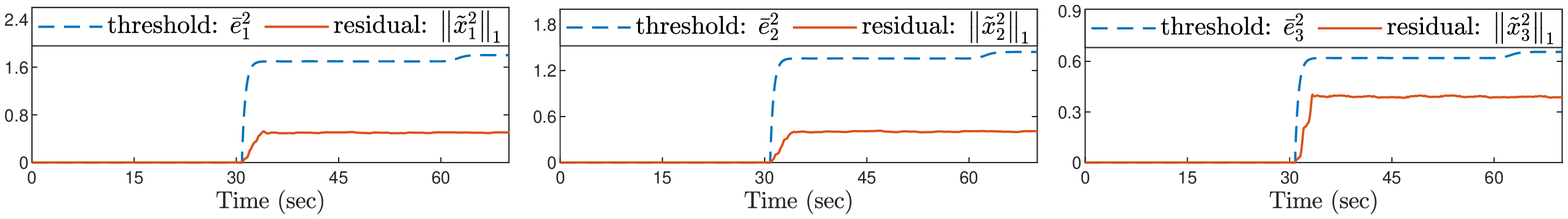}	
	\caption{}
	\label{4_fault2-FI2}
	\end{subfigure}
	\begin{subfigure}[t]{0.85 \textwidth}
	\includegraphics[width=1 \textwidth]{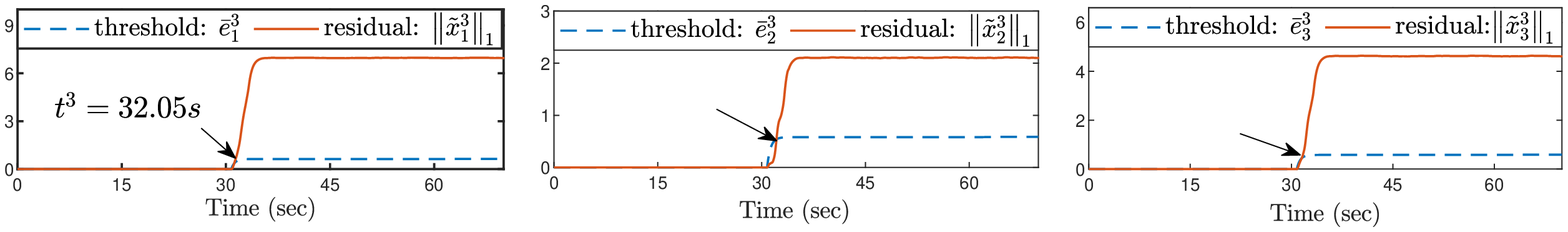}	
	\caption{}
	\label{4_fault2-FI3}
	\end{subfigure}
%	\caption{ FDI performance when fault $2'$ occurs in system (\ref{SIM_sys}). Time profiles of (a) the FD residuals $\left\| \tilde{x}_i^0 (t) \right\|_1$ and FD thresholds $\bar{e}_i^0 $; (b) the $1$-st FI residuals $\left\| \tilde{x}_i^1 (t) \right\|_1$ and FI thresholds $\bar{e}_i^1 (t)$; (c) the  $2$-nd FI residuals $\left\| \tilde{x}_i^2 (t) \right\|_1$ and FI thresholds $\bar{e}_i^2 (t)$; and (d) the $3$-rd FI residuals $\left\| \tilde{x}_i^3 (t) \right\|_1$ and FI thresholds $\bar{e}_i^3 (t)$, $\forall i=1,2,3$. }
	\caption{ FDI performance when fault $2'$ occurs at time $t_0=30s$ in system (\ref{SIM_sys}): (a) FD residuals and thresholds; (b) $1$-st FI residuals and thresholds; (c) $2$-nd FI residuals and thresholds; and (d) $3$-rd FI residuals and thresholds. }
	\label{PLOT_FD-fault2m}
\end{figure*}

%\begin{figure*} [htb!]
%\centering
%	\includegraphics[width= 0.85 \textwidth]{pics/5_fault3m.eps}	
%	\caption{ FDI performance when fault $3'$ occurs in system (\ref{SIM_sys}). Time profiles of (a) the FD residuals $\left\| \tilde{x}_i^0 (t) \right\|_1$ and FD thresholds $\bar{e}_i^0$; (b) the $1$-st FI residuals $\left\| \tilde{x}_i^1 (t) \right\|_1$ and FI thresholds $\bar{e}_i^1 (t)$; (c) the $2$-nd FI residuals $\left\| \tilde{x}_i^2 (t) \right\|_1$ and FI thresholds $\bar{e}_i^2 (t)$; and (d) the $3$-rd FI residuals $\left\| \tilde{x}_i^3 (t) \right\|_1$ and FI thresholds $\bar{e}_i^3 (t)$,  $\forall i=1,2,3$. }
%	\label{PLOT_FD-fault3m}
%\end{figure*}

\begin{figure*} [htb!]
\centering
	\begin{subfigure}[t]{0.85 \textwidth}
	\includegraphics[width=1 \textwidth]{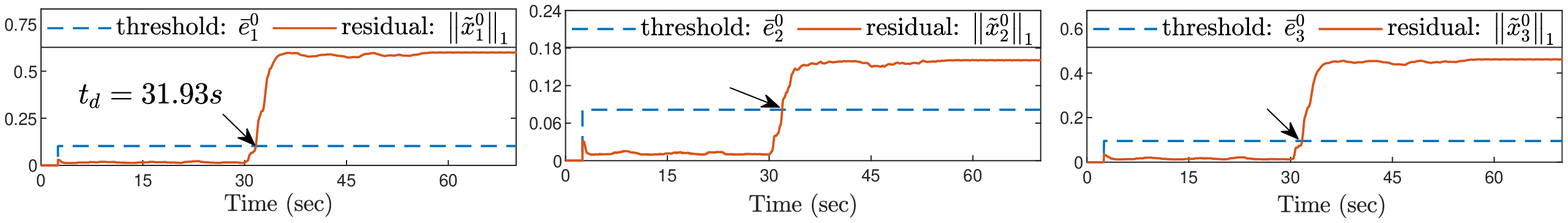}	
	\caption{}
	\label{5_fault3-FD}
	\end{subfigure}
	\begin{subfigure}[t]{0.85 \textwidth}
	\includegraphics[width=1 \textwidth]{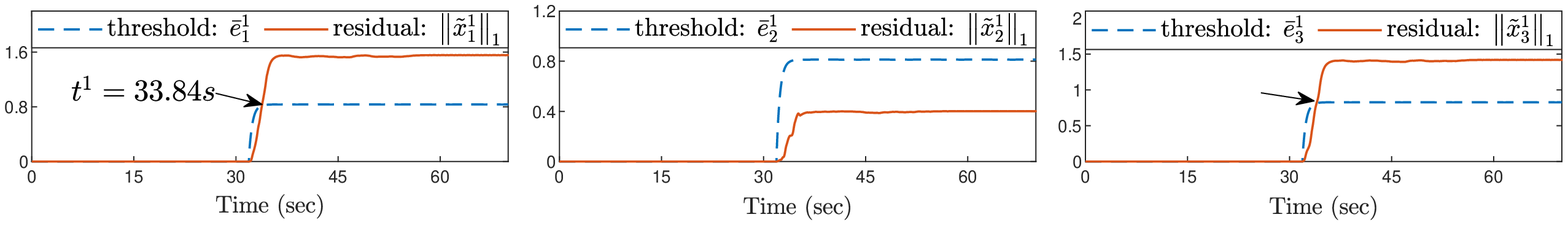}	
	\caption{}
	\label{5_fault3-FI1}
	\end{subfigure}
	\begin{subfigure}[t]{0.85 \textwidth}
	\includegraphics[width=1 \textwidth]{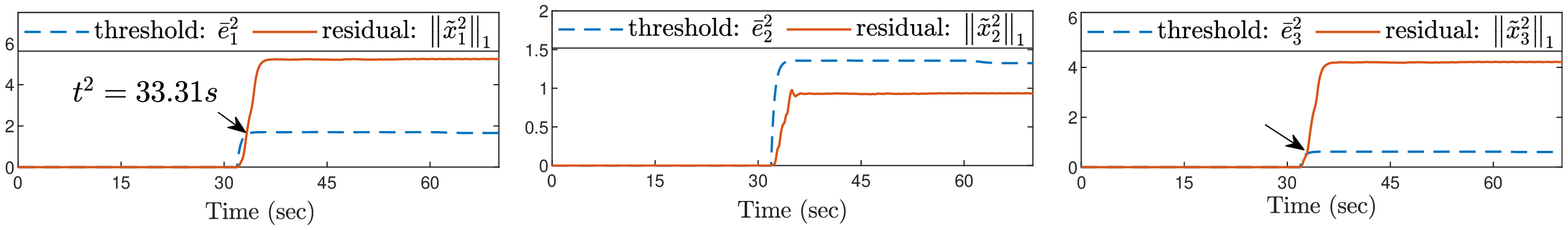}	
	\caption{}
	\label{5_fault3-FI2}
	\end{subfigure}
	\begin{subfigure}[t]{0.85 \textwidth}
	\includegraphics[width=1 \textwidth]{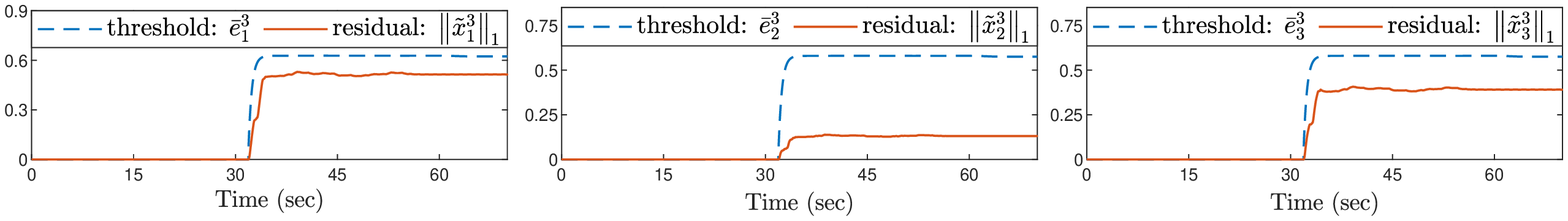}	
	\caption{}
	\label{5_fault3-FI3}
	\end{subfigure}
	%\caption{ FDI performance when fault $3'$ occurs in system (\ref{SIM_sys}). Time profiles of (a) the FD residuals $\left\| \tilde{x}_i^0 (t) \right\|_1$ and FD thresholds $\bar{e}_i^0$; (b) the $1$-st FI residuals $\left\| \tilde{x}_i^1 (t) \right\|_1$ and FI thresholds $\bar{e}_i^1 (t)$; (c) the $2$-nd FI residuals $\left\| \tilde{x}_i^2 (t) \right\|_1$ and FI thresholds $\bar{e}_i^2 (t)$; and (d) the $3$-rd FI residuals $\left\| \tilde{x}_i^3 (t) \right\|_1$ and FI thresholds $\bar{e}_i^3 (t)$,  $\forall i=1,2,3$. }
	\caption{ FDI performance when fault $3'$ occurs at time $t_0=30s$ in system (\ref{SIM_sys}): (a) FD residuals and thresholds; (b) $1$-st FI residuals and thresholds; (c) $2$-nd FI residuals and thresholds; and (d) $3$-rd FI residuals and thresholds. }
	\label{PLOT_FD-fault3m}
\end{figure*}

\section{Simulation Studies} \label{sec_simulation}
Consider a typical transport-reaction process in chemical industry, i.e., a long, thin catalytic rod in a reactor, which is borrowed from \cite{ElFG.AIChE07, WuL.TNN08}. The spatio-temporal evolution of the rod temperature is described by the following parabolic PDE:
\bgeq \label{SIM_sys}
\begin{aligned}
	\frac{\partial x(z,t)}{\partial t} = \,& \frac{\partial^2 x(z,t)}{\partial z^2}+ \beta_T (e^{-\frac{\gamma}{1+x}}-e^{-\gamma}) \\
	& + \beta_u(b(z) u(t) - x(z,t)) + \beta(t-t_0)\phi^k(x,u), 
\end{aligned}
\edeq
with boundary and initial conditions: $x(0,t) =0, \, x(\pi,t)=0, \, x(z,0) = 15 \sin(z)$, where $x(z,t)$ denotes the rod temperature, $u(t)$ is the manipulated input, $f(x,u) = \beta_T (e^{-\frac{\gamma}{1+x}}-e^{-\gamma})+ \beta_u(b(z) u - x) $ is an unknown function representing the system uncertainty, and $\phi^k(x,u)$ is the fault function. $\beta_T=50$ is a heat of action, $\gamma=4$ is an activation energy, $\beta_u=2$ is a heat transfer coefficient, $b(z)=1.5 \sin(z)+1.8 \sin(2z) + 2\sin(3z)$ is the actuator distribution function, $t_0=30s$ is the fault occurrence time. 
For simulation purpose, the system input is set as $u(t) = 1.1+2 \sin(5t) - 2 \cos(5t)$. Three types of faults are considered. 
(i) Fault 1: an actuator fault with a faulty actuator distribution function:  $b'(z) = 1.8 \sin(z)+1.8 \sin(2z) + 2\sin(3z)$, leading to the associated fault function $\phi^1(x,u) = \tilde{ b}(z) \beta_u u $ with $\tilde{ b}(z) = b(z)-b'(z) = - 0.3\sin(z)$. 
(ii) Fault 2: a state fault with fault function $\phi^2(x,u)= \tilde{h}(z) x$, where $\tilde{h}(z) = h(z-1)-h(z-1.3)$ and $h(\cdot)$ is a heaviside function.   
(iii) Fault 3: a component fault with a faulty system parameter: $\beta_T' = 48$, and the associated fault function is thus $\phi^3(x,u)= \tilde{ \beta_T} (e^{-\frac{\gamma}{1+x}}-e^{-\gamma})$ with $\tilde{\beta}_T = \beta_T -\beta_T'=2$. 
%
%At most one fault will be occurring at each time instant. 

For the PDE system (\ref{SIM_sys}), we first derive its approximate ODE model. % in the form of (\ref{ODE_sys}). % via Galerkin method. 
Specifically, the eigenvalue problem of the spatial differential operator in (\ref{SIM_sys}), i.e., $\mathcal{A}x = \frac{\partial^2 x}{\partial z^2}, \, x \in D(\mathcal{A}):= \left\{x \in \mathcal{H} \,|\,  x(0,t) =0, \, x(\pi,t)=0 \right\}$,  
can be solved analytically by using the method of \cite{DenLC.TCST05}, resulting in the solution: $\lambda_i = -i^2, \, \varphi_i(z) = \sqrt{\frac{2}{\pi}} \sin (iz) $ with $ i = 1,\cdots, \infty$. By choosing the first $m=3$ number of eigenvalues as dominant ones, we can obtain the following ODE system to describe the dominant dynamics of the PDE system (\ref{SIM_sys}): 
\bgeq \label{SIM_ODE-sys}
\begin{aligned}
	\dot{x}_{s_i} = \lambda_i x_{s_i} + f_{s_i}(x_s,u) + &\, \beta(t-t_0) \phi_{s_i}^k(x_s,u), \\
%	&\, i=1,2,3, \quad k=1,2,3,
\end{aligned}
\edeq
where $i=1,2,3,\,k=1,2,3$, $ x_{s_i}(t) = \int_0^{\pi} x(z,t) \varphi_i(z) \, dz $, $f_{s_i}(x_{s},u) = - \beta_ux_{s_i}(t)+ \int_0^{\pi} (\beta_T (e^{-\frac{\gamma}{1+ \sum_{i=1}^3 x_{s_i}(t) \varphi_i(z) }}-e^{-\gamma}) + \beta_u b(z) u(t) ) \varphi_i(z) \, dz$, $ \phi^1_{s_i}(x_{s}, u) = \beta_u u(t) \int_0^{\pi} \tilde{ b}(z) \varphi_i(z) \, dz$, $ \phi^2_{s_i}(x_{s}, u)= \int_{0}^{\pi} x(z,t) \tilde{ h}(z) \varphi_i(z) \, dz $, and $ \phi^3_{s_i}(x_{s}, u) = \tilde{\beta_T} \int_{0}^{\pi} (e^{-\frac{\gamma}{1+ \sum_{i=1}^3 x_{s_i}(t) \varphi_i(z) }}-e^{-\gamma}  ) \varphi_i(z) \, dz $. 
%
%The system of (\ref{SIM_ODE-sys}) can be used to describe the dominant dynamics of PDE system (\ref{SIM_sys}). 
Note that the model (\ref{SIM_ODE-sys}) cannot be directly used for the design of FDI scheme, due to the existence of uncertain functions $f_{s_i}(x_s, u)$ and $\phi^k_{s_i}(x_s, u)$. 
The state signals $x_{s_i}$, which are needed for the subsequent implementation, will be obtained based on the measurement from the original PDE system (\ref{SIM_sys}) via $ x_{s_i}(t) = \int_0^{\pi} x(z,t) \varphi_i(z) \, dz $, as discussed in Remark \ref{rem_identi-xs}. 

With the system signals $(x_s,u)$, we can implement the identification process for the uncertain dynamics $\eta^k_i(x_s,u)= f_{s_i}(x_s,u)+\phi^k_{s_i}(x_s,u)$ ($\forall i =1,2,3$) of system (\ref{SIM_ODE-sys}) under all normal and faulty modes with $k=0,1,2,3$. 
Specifically, according to (\ref{Identify_obs}), the RBF network $\hat{W}_i^{k\top}S(x_s,u)$ is constructed in a regular lattice, with nodes $N_n = 14 \times 9 \times 8 \times 13$, the center evenly spaced on $[17.5,24] \times [-1,3] \times [0,3.5] \times [-2,4]$ and the widths $\nu_i = 0.5$. The design parameters in  (\ref{Identify_obs}) are $a_i = 4$, $\Gamma_i= 0.35$ and $\sigma_i=0.001$ ($\forall i=1,2,3$). The initial conditions are set as $\hat{W}_i^k(0)=0$ and $\hat{x}(0) = x_s(0)$. % $\hat{x}(0) = [18.7;0;0]$. 
Consider the system (\ref{SIM_ODE-sys}) operating in normal mode (with $k=0$, $\phi^0_{s_i}(x_s,u) \equiv 0$), with identifier (\ref{Identify_obs}), the identification performance for the dynamics of the $1$-st state subsystem of (\ref{SIM_ODE-sys}) is shown in Fig. \ref{PLOT_identify-0}. Particularly, Fig. \ref{1a_weight} shows the convergence of NN weight $\hat{W}_1^0$. Fig. \ref{1b_xhat} shows the accurate tracking performance of $\hat{x}_1$ over the system state $x_{s_1} $. %$ = \int_0^{\pi} x(z,t) \varphi_1(z) \, dz $. 
Based on the identification result, a constant model $\bar{W}^{0\top}_1S(x_s,u)$ is obtained by $\bar{W}_1^0 = \frac{1}{10} \int_{140}^{150} \hat{W}_1^0 (\tau) \, d \tau$, which can achieve accurate approximation of the associated unknown function $\eta_{1}^0(x_s,u)$, as illustrated in Fig. \ref{1d_func}. 
%
%Next, consider the system (\ref{SIM_ODE-sys}) operating under faulty mode $k=1$, and the identification performance for the $1$-st state subsystem can be seen in Fig. \ref{PLOT_identify-1}. It shows similar results and verifies that accurate approximation of $\eta_1^1(x_s,u)$ can be achieved by $\bar{W}^{1\top}_1S(x_s,u)$ (where $\bar{W}_1^1 = \frac{1}{10} \int_{140}^{150} \hat{W}_1^1 (\tau) \, d \tau$).  
%
Then, following a similar procedure established as above, simulation results for the cases of faulty modes $k=1,2,3$ can also be obtained, which are similar to those in Fig. \ref{PLOT_identify-0} and thus omitted here. 
Consequently, with the method given in Remark \ref{rem_epsilon*}, 
the values of $\xi_i^*$ ($i=1,2,3$) that are needed for implementing the subsequent FDI scheme can be obtained as $\xi_1^*=0.0860$, $\xi_2^*=0.0430$, and $\xi_3^*=0.0703$. 

Based on the above identification results, we can implement the proposed FDI scheme for system (\ref{SIM_ODE-sys}).  
Specifically, the FD estimators (\ref{FD_estimator}) are implemented  with  constant NN models $\bar{W}_i^{0\top}S(x_s,u)$ and parameters $b_i^0 =2$ ($\forall i=1,2,3$). The FD thresholds (\ref{FD_threshold}) are implemented with parameters $\xi_1^*=0.0860$, $\xi_2^*=0.0430$, $\xi_3^*=0.0703$, and $\varrho_i =0.12$. The parameter of $\mathcal{L}_1$ norm is set as $T=2.5s$. 
Similarly, the FI estimators (\ref{FI_estimator}) are implemented with constant NN models $\bar{W}_i^{k\top}S(x_s,u)$ and parameters $b_i =2$ ($\forall i=1,2,3$ and $\forall k=1,2,3$). 
The FI adaptive thresholds (\ref{FI_threshold}) are implemented with given functions $\bar{\rho}^k_i (x_s, u) = \int_{0}^{\pi} \bar{\phi}^k(x,u) \left| \varphi_i(z)\right| d z$, where $\bar{\phi}^1(x,u) = \left| \Delta_{\tilde{ b}} \beta_u u \right|$ with $\Delta_{\tilde{ b}} = 0.25$, $\bar{\phi}^2(x,u) = \left| \Delta_{\tilde{ h}} x \right|$ with $\Delta_{\tilde{ h}} = h(z-1) - h(z-1.3)$, and $\bar{\phi}^3(x,u) = \left| \Delta_{\tilde{ \beta_T}} (e^{-\frac{\gamma}{1+x}}-e^{-\gamma}) \right|$ with $\Delta_{\tilde{\beta_T}} = 1$. 
For testing purpose, we assume three occurring faults to be detected and isolated, including fault $1'$: $\phi^{1'}(x,u) = \tilde{b}'(z) \beta_u u $ with $\tilde{b}'(z) = -0.5\sin(z) $; fault $2'$: $\phi^{2'}(x,u) = \tilde{ h}'(z) x$ with $\tilde{ h}'(z) = h(z-1)-h(z-1.2)$; and fault $3'$: $\phi^{3'}(x,u) = \tilde{ \beta_T} ' (e^{-\frac{\gamma}{1+x}} - e^{-\gamma})$ with $\tilde{ \beta_T}' = 49$. These faults satisfy $ \left| \phi^k(x,u) - \phi^{k'}(x,u) \right| \leq \bar{\phi}^k(x,u)$ and $\left| \rho^{k,k'}_{i}(x_s,u) \right| = \left| \int_{0}^{\pi} ( \phi^{k'}(x,u) - \phi^{k}(x,u) ) \varphi_i(z) dz \right| \leq  \int_{0}^{\pi} \bar{\phi}^k(x,u) \left| \varphi_i(z)\right| d z = \bar{\rho}^k_i (x_s, u)$ 
%
%$\left| \rho^{k,k'}_{i}(x_s,u) \right| = \left| \phi^k_{s_i}(x_s,u) - \phi^{k'}_{s_i}(x_s,u) \right| = \left| \int_{0}^{\pi} ( \phi^{k'}(x,u) - \phi^{k}(x,u) ) \varphi_i(z) dz \right| \leq  \int_{0}^{\pi} \bar{\phi}^k(x,u) \left| \varphi_i(z)\right| d z = \bar{\rho}^k_i (x_s, u)$ 
for all $k=1,2,3$ and $i=1,2,3$, which verifies Assumption \ref{ASS_fault-bound}. 

In the testing phase, consider the fault $1'$ occurring in system (\ref{SIM_sys}) at time $t_0=30s$, the associated FDI simulation results are displayed in Fig. \ref{PLOT_FD-fault1m}. 
We first observe the FD performance in Fig. \ref{3_fault1-FD}. It is shown that once the fault $1'$ occurs at time $t_0=30s$, all FD residuals $\left\| \tilde{x}_i^0 \right\|_1$ ($i=1,2,3$) increase and become larger than the associated thresholds $\bar{e}_i^0$ at time $t_d = 30.9s$, indicating that the occurring fault $1'$ can be detected at time $t_d = 30.9s$. %, i.e., $0.9s$ after its occurrence. 
Once fault $1'$ is detected, the FI system consisting of FI estimators (\ref{FI_estimator}) and FI adaptive thresholds (\ref{FI_threshold}) is activated, and the FI performance can be seen in Figs. \ref{3_fault1-FI1}-\ref{3_fault1-FI3}. 
For the performance of the matched/similar FI estimator (i.e., the $1$-st FI estimator) as shown in Fig. \ref{3_fault1-FI1}, it is seen that all the residual signals $\left\| \tilde{x}_i^1(t) \right\|_1$ ($i=1,2,3$) remain smaller than the associated threshold $\bar{e}_i^1 (t) $ for all time $t > t_d = 30.9s$. 
For those estimators with unmatched/unsimilar faults (i.e., the $2$-nd and the $3$-rd FI estimators), the associated performance is presented in Figs. \ref{3_fault1-FI2}, \ref{3_fault1-FI3}, respectively. It is shown that the $2$-nd FI residual $\left\| \tilde{x}_i^2 (t) \right\|_1$ (with $i=3$) becomes larger than the threshold $\bar{e}_i^2 (t)$ at time $ t^2=32.06s$ (see Fig. \ref{3_fault1-FI2}); and all the $3$-rd FI residuals $\left\| \tilde{x}_i^3(t) \right\|_1$ (with all $i=1,2,3)$ become larger than the respective thresholds $\bar{e}_i^3(t)$ at time $ t^3= 32.03s$ (see Fig. \ref{3_fault1-FI3}). 
Thus, it can be deduced that the occurring fault $1'$ is similar to the fault $1$, and the isolation time can be obtained at: $t_{iso} = \max \{ t^2, t^3 \} = 32.06s$.
Next, we further consider the cases when faults $2'$ and $3'$ are occurring respectively in system (\ref{SIM_sys}) at time $t_0=30s$, and the associated FDI performances are illustrated in Figs. \ref{PLOT_FD-fault2m} and \ref{PLOT_FD-fault3m}, respectively. It is seen that the occurring fault $2'$ is detected at time $t_d = 30.85s$ and isolated at time $t_{iso} = 32.52s$; while the occurring fault $3'$ is detected at time $t_d = 31.93s$ and isolated at time $t_{iso} = 33.84s$. 
These simulation results demonstrate feasibility and effectiveness of our proposed  FDI scheme.

\begin{figure*} [htb!]
\centering
	\begin{subfigure}[t]{0.32 \textwidth}
	\includegraphics[width=1.05 \textwidth]{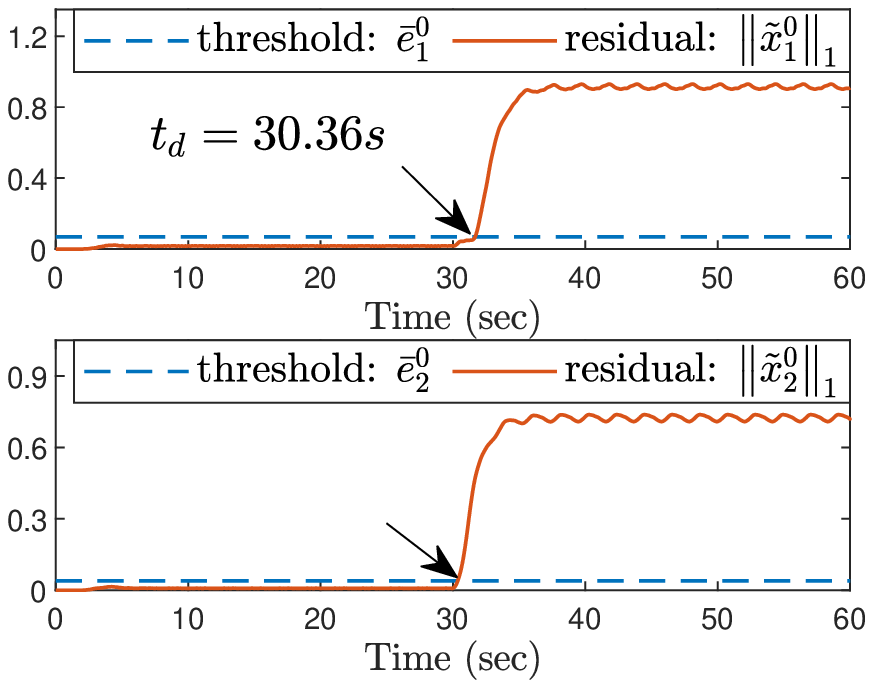}	
	\caption{}
	\label{6_mine-FD}
	\end{subfigure}
	\begin{subfigure}[t]{0.32 \textwidth}
	\includegraphics[width=1.05 \textwidth]{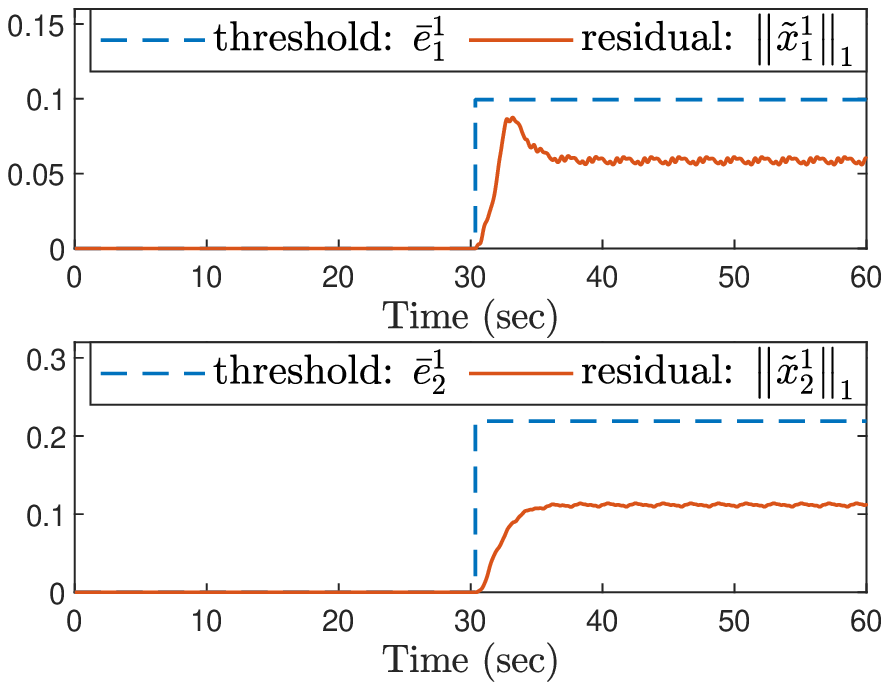}	
	\caption{}
	\label{6_mine-FI1}
	\end{subfigure}
	\begin{subfigure}[t]{0.32 \textwidth}
	\includegraphics[width=1.05 \textwidth]{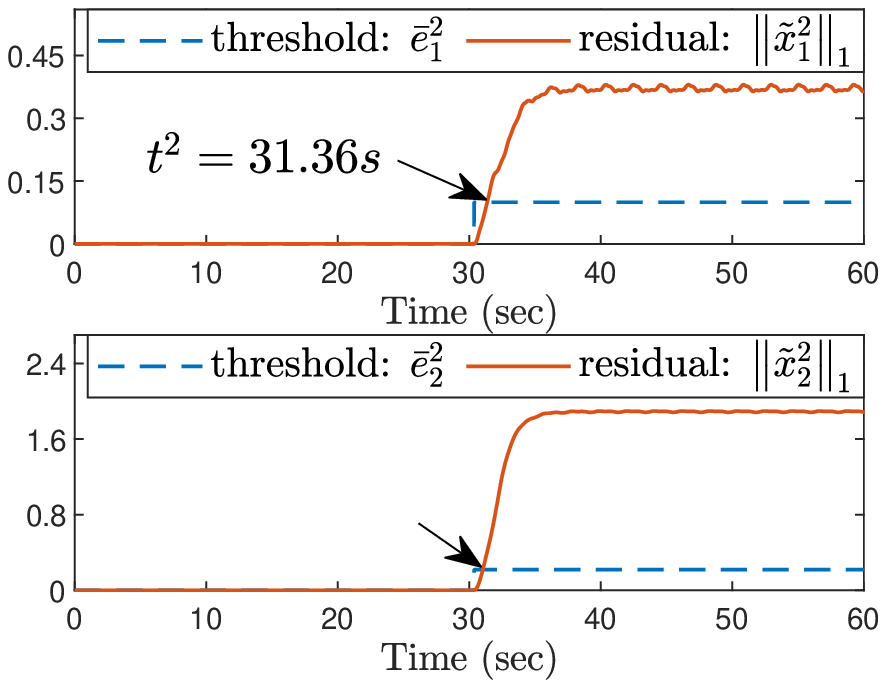}	
	\caption{}
	\label{6_mine-FI2}
	\end{subfigure}
	\caption{ FDI performance with the proposed scheme when actuator fault $1'$ occurs at time $t_0=30s$ (for comparison study): (a) FD residuals and thresholds; (b) $1$-st FI residuals and thresholds; and (c) $2$-nd FI residuals and thresholds. }
	\label{PLOT_compM-fault1}
\end{figure*}

\begin{figure*} [htb!]
\centering
	\begin{subfigure}[t]{0.37 \textwidth}
	\includegraphics[width=1 \textwidth]{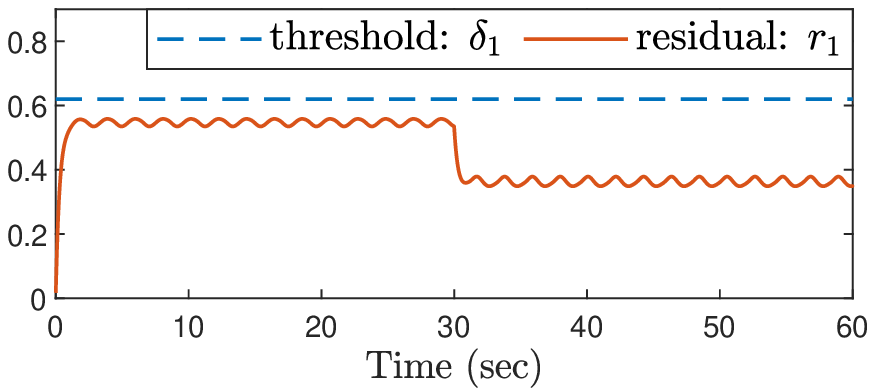}	
	\caption{}
	\label{6_other-FI1}
	\end{subfigure}
	\begin{subfigure}[t]{0.37 \textwidth}
	\includegraphics[width=1 \textwidth]{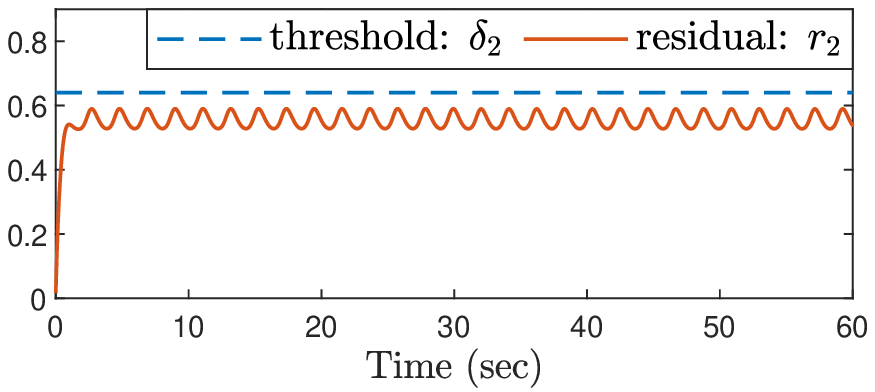}	
	\caption{}
	\label{6_other-FI2}
	\end{subfigure}
	\caption{ FI performance with the method of \cite{ElFG.AIChE07} when actuator fault $1'$ occurs at time $t_0=30s$ (for comparison study): (a) $1$-st FI residual and threshold; and (b) $2$-nd FI residual and threshold. }
	\label{PLOT_compO-fault1}
\end{figure*}

\begin{figure*} [htb!]
\centering
	\begin{subfigure}[t]{0.32 \textwidth}
	\includegraphics[width=1.05 \textwidth]{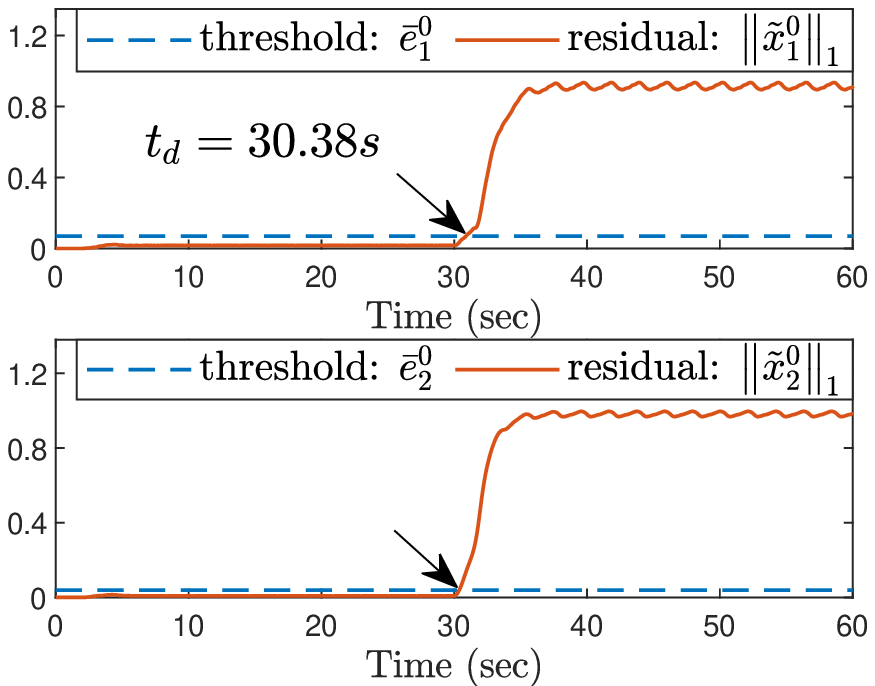}	
	\caption{}
	\label{7_mine-FD}
	\end{subfigure}
	\begin{subfigure}[t]{0.32 \textwidth}
	\includegraphics[width=1.05 \textwidth]{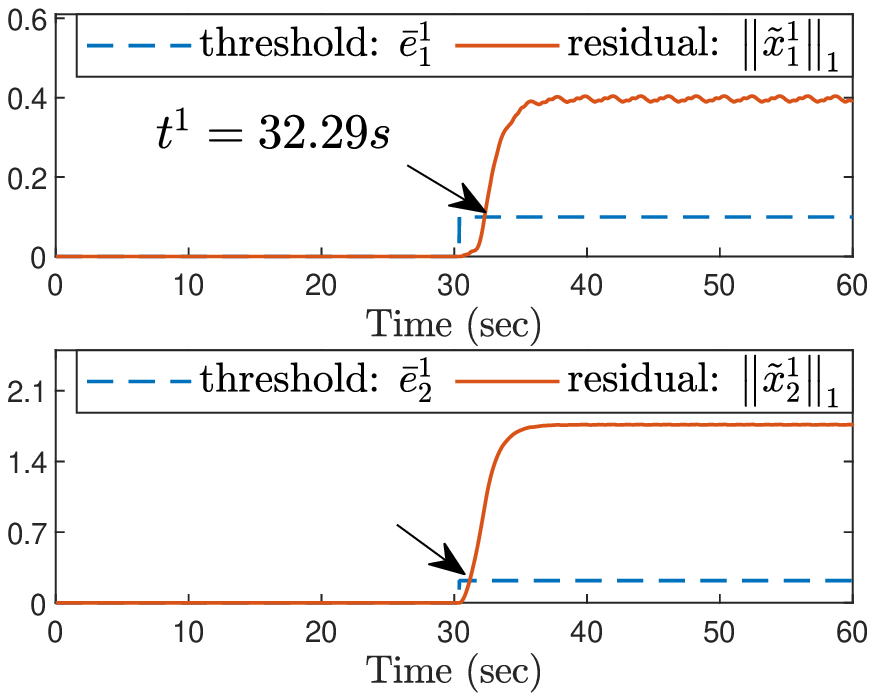}	
	\caption{}
	\label{7_mine-FI1}
	\end{subfigure}
	\begin{subfigure}[t]{0.32 \textwidth}
	\includegraphics[width=1.05 \textwidth]{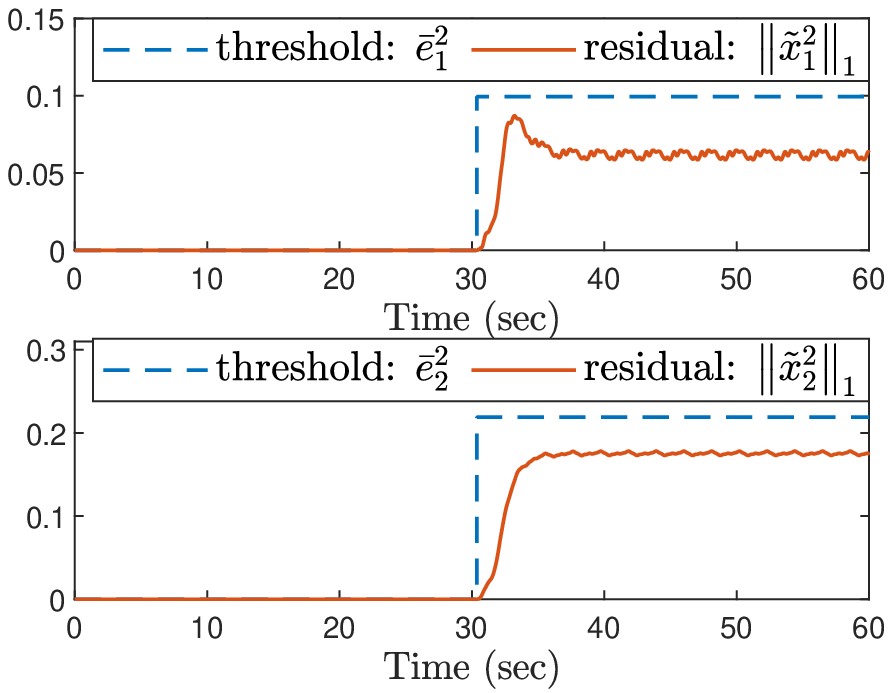}	
	\caption{}
	\label{7_mine-FI2}
	\end{subfigure}
	\caption{ FDI performance with the proposed scheme when actuator fault $2'$ occurs at time $t_0=30s$ (for comparison study): (a) FD residuals and thresholds; (b) $1$-st FI residuals and thresholds; and (c) $2$-nd FI residuals and thresholds. }
	\label{PLOT_compM-fault2}
\end{figure*}

\begin{figure*} [htb!]
\centering
	\begin{subfigure}[t]{0.37 \textwidth}
	\includegraphics[width=1 \textwidth]{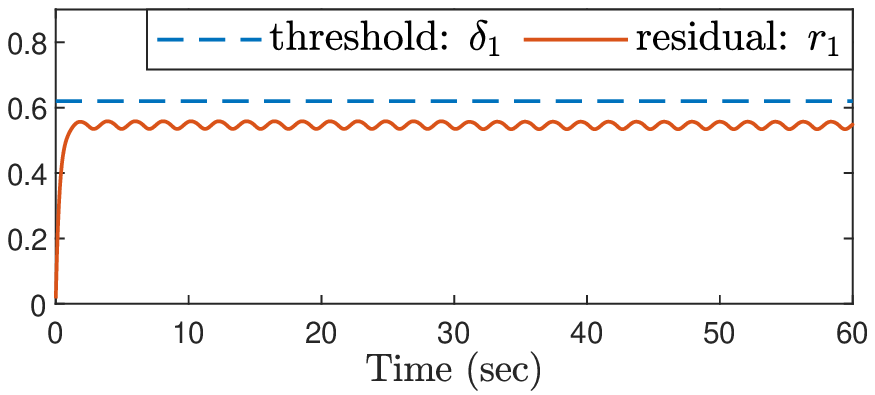}	
	\caption{}
	\label{7_other-FI1}
	\end{subfigure}
	\begin{subfigure}[t]{0.37 \textwidth}
	\includegraphics[width=1 \textwidth]{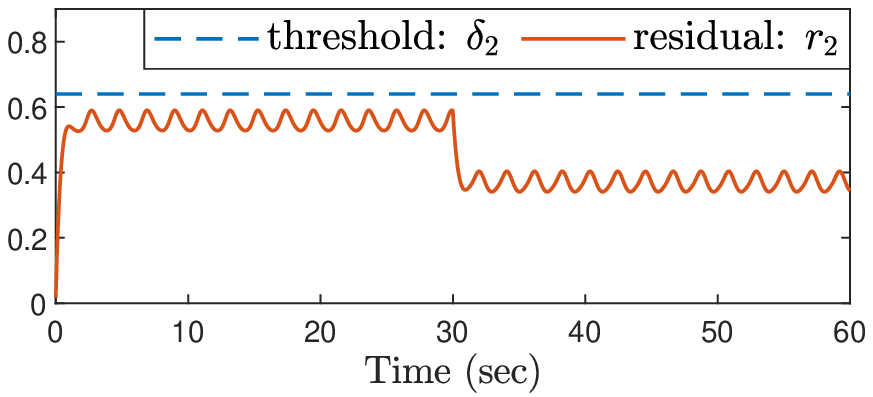}	
	\caption{}
	\label{7_other-FI2}
	\end{subfigure}
	\caption{ FI performance with the method of \cite{ElFG.AIChE07} when actuator fault $2'$ occurs at time $t_0=30s$ (for comparison study): (a) $1$-st FI residual and threshold; and (b) $2$-nd FI residual and threshold. }
	\label{PLOT_compO-fault2}
\end{figure*}

To further justify the advantage of our FI scheme in dealing with the system uncertainty for accurate isolation, we compare the performance of our scheme with the existing method in \cite{ElFG.AIChE07}. %\cite{ElFG.AIChE07, BanK.ACC12, DemA.JRNC12}. 
%
%Particularly, comparison study between our FD scheme and the one in \cite{DemA.JRNC12} has been performed in our previous work \cite{ZhaY.CSL20} thus is omitted here. In the following, simulation study will be conducted by comparing our FI scheme and the ones in \cite{ElFG.AIChE07}. To this end, 
%
A PDE system in the form of (\ref{SIM_sys}) is considered, where the system structures/parameters keep unchanged except $b(z)=[h(z)-h(z-\pi/2),h(z-\pi/2)-h(z-\pi)]$ and $u(t)=[1.1+6\sin(3t); 1.1-6\cos(3t)]$. 
This system model is assumed partially-unknown, i.e., it consists of (known) nominal component $\frac{\partial^2 x}{\partial z^2}+f(x)+\beta_u b(z)u=\frac{\partial^2 x}{\partial z^2}+0.9\beta_T(e^{-\frac{\gamma}{1+x}}-e^{-\gamma})+\beta_u (b(z)u-x)$ and uncertain/unknown component $N(x)=0.1\beta_T(e^{-\frac{\gamma}{1+x}}-e^{-\gamma})$. 
Two types of actuator faults at different locations are considered, i.e., $\phi^k(x,u)=\beta_u b(z)f_a^k(x)$ ($k=1,2$) with $f_a^1(x)=[-0.05x(\frac{\pi}{2},t);0]$ denoting fault 1 and $f_a^2(x)=[0;-0.05x(\frac{\pi}{2},t)]$ denoting fault 2. 
The approximate ODE model of this system is derived with order $m=2$. 
For the FI scheme of \cite{ElFG.AIChE07}, a bank of FI filters (generating FI residuals $r_i(t)$, $i=1,2$) are constructed according to \cite[Eq. (27)]{ElFG.AIChE07} and the corresponding FI thresholds are given as $\delta_1=0.62$ and $\delta_2= 0.64$, which are determined based on the upper bound of system uncertainty $N(x)$, according to \cite[Remark 19]{ElFG.AIChE07}.  
For our scheme, the implementation process follows a similar line established as above, in which the RBF network is constructed with nodes $N_n = 13 \times 9 \times 15 \times 15$, the center evenly spaced on $[14,26] \times [-4,4] \times [-6,8] \times [-6,8]$ and the widths $\nu_i = 1$; and the associated parameters are given as $b_i^0=b_i=1$, $\varrho_i=0.02$ ($i=1,2$), $\xi_1^*=0.0495$, $\xi_2^*=0.0191$, $\bar{\rho}_1^1=0.05$, $\bar{\rho}_1^2=0.2$ and $T=2s$.  
For testing purpose, two test faults, i.e., fault $1'$ with $f_a^{1'}(x)=[-0.043x(\frac{\pi}{2},t);0]$ and fault $2'$ with $f_a^{2'}(x)=[0;-0.043x(\frac{\pi}{2},t)]$, are considered occurring at time $t_0=30s$.  
Specifically, considering the case when fault $1'$ occurs, with our scheme, it can be seen in Fig. \ref{PLOT_compM-fault1} that the occurring fault $1'$ can be detected at time $t_d=30.36s$ and be identified similar to fault $1$ at time $t_{iso}=t^2=31.36s$. 
With the scheme of \cite{ElFG.AIChE07}, it is shown in Fig. \ref{PLOT_compO-fault1} that after fault occurrence time $t_0=30s$, the matched FI residual $r_1(t)$ do not increase and cross the associated FI threshold $\delta_{1}$, indicating that isolation for fault $1'$ cannot be achieved.  
%
%Similar results for the case of fault $2'$ can be seen in Figs. \ref{PLOT_compM-fault2}--\ref{PLOT_compO-fault2}. 
For the case of fault $2'$, similar observations can be seen in Figs. \ref{PLOT_compM-fault2}--\ref{PLOT_compO-fault2}, where fault $2'$ can be detected at $t_d=30.38s$ and isolated at $t_{iso}=32.29s$ with our scheme; but isolation failed with the scheme of \cite{ElFG.AIChE07}. 
For such results, one important reason lies in that: the FI method of \cite{ElFG.AIChE07} cannot deal with the effect of the system uncertainty $N(x)$ during the FI process, such that the occurring fault dynamics $\phi^k(x,u)$ are hidden within the uncertain dynamics $N(x)$ and cannot be captured for successful isolation; while our method has successfully overcome this issue by achieving accurate identification of system uncertainty $N(x)$. % and accurately measuring the fault dynamics for FI. 
These comparison results are consistent with the discussions in Remarks \ref{rem_FDestimator} and \ref{rem_FIestimator}, demonstrating the advantage of our FI scheme compared to that of \cite{ElFG.AIChE07}.

\begin{remark}
Comparison study for the proposed FD scheme has been performed in our preliminary  work \cite{ZhaY.CSL20}, which is thus not repeated here.  
\end{remark}

\section{Conclusions} \label{sec_conclusion}

In this paper, we have proposed a novel FDI scheme for a class of uncertain nonlinear parabolic PDE systems. 
The design was based on an approximate ODE system derived via the Galerkin method, which is used to capture the dominant dynamics of the original PDE system. 
Specifically, based on the ODE system, a DL-based adaptive dynamics learning approach was first developed to achieve locally-accurate identification of the system uncertain dynamics under normal and all faulty modes. The learned knowledge was obtained and stored in constant RBF NN models. 
Then, a bank of FDI estimators can be designed with these models. In particular, the FD estimators are used to detect the occurrence of a fault; while the FI estimators, which will be activated once the fault is detected, are used to identify the type of occurring fault. 
The thresholds associated with these estimators were further designed for real-time decision making. 
The associated analysis on FDI performance, i.e., fault detectability and isolatability conditions, has also been provided. 
Simulation studies have been conducted to verify the effectiveness and advantage of the proposed methodologies. 


\begin{thebibliography}{10}
\providecommand{\url}[1]{#1}
\csname url@samestyle\endcsname
\providecommand{\newblock}{\relax}
\providecommand{\bibinfo}[2]{#2}
\providecommand{\BIBentrySTDinterwordspacing}{\spaceskip=0pt\relax}
\providecommand{\BIBentryALTinterwordstretchfactor}{4}
\providecommand{\BIBentryALTinterwordspacing}{\spaceskip=\fontdimen2\font plus
\BIBentryALTinterwordstretchfactor\fontdimen3\font minus
  \fontdimen4\font\relax}
\providecommand{\BIBforeignlanguage}[2]{{%
\expandafter\ifx\csname l@#1\endcsname\relax
\typeout{** WARNING: IEEEtran.bst: No hyphenation pattern has been}%
\typeout{** loaded for the language `#1'. Using the pattern for}%
\typeout{** the default language instead.}%
\else
\language=\csname l@#1\endcsname
\fi
#2}}
\providecommand{\BIBdecl}{\relax}
\BIBdecl

\bibitem{MonY.CNSNS12}
G.~Montaseri and M.~J. Yazdanpanah, ``Predictive control of uncertain nonlinear
  parabolic pde systems using a galerkin/neural-network-based model,''
  \emph{Communications in Nonlinear Science and Numerical Simulation}, vol.~17,
  no.~1, pp. 388--404, 2012.

\bibitem{HerM.IS99}
B.~Hernandez-Morales and A.~Mitchell, ``Review of mathematical models of fluid
  flow, heat transfer, and mass transfer in electroslag remelting process,''
  \emph{Ironmaking \& steelmaking}, vol.~26, no.~6, pp. 423--438, 1999.

\bibitem{BouCK.SP15}
S.~Bououden, M.~Chadli, and H.~R. Karimi, ``Control of uncertain highly
  nonlinear biological process based on takagi--sugeno fuzzy models,''
  \emph{Signal Processing}, vol. 108, pp. 195--205, 2015.

\bibitem{ElFG.AIChE07}
N.~H. El-Farra and S.~Ghantasala, ``Actuator fault isolation and
  reconfiguration in transport-reaction processes,'' \emph{AIChE Journal},
  vol.~53, no.~6, pp. 1518--1537, 2007.

\bibitem{LuZH.TII16}
X.~Lu, W.~Zou, and M.~Huang, ``A novel spatiotemporal ls-svm method for complex
  distributed parameter systems with applications to curing thermal process,''
  \emph{IEEE Transactions on Industrial Informatics}, vol.~12, no.~3, pp.
  1156--1165, 2016.

\bibitem{RusCB.SSBM12}
E.~L. Russell, L.~H. Chiang, and R.~D. Braatz, \emph{Data-driven methods for
  fault detection and diagnosis in chemical processes}.\hskip 1em plus 0.5em
  minus 0.4em\relax London, UK: Springer, 2012.

\bibitem{DemA.JRNC12}
M.~Demetriou and A.~Armaou, ``Dynamic online nonlinear robust detection and
  accommodation of incipient component faults for nonlinear dissipative
  distributed processes,'' \emph{International Journal of Robust and Nonlinear
  Control}, vol.~22, no.~1, pp. 3--23, 2012.

\bibitem{DemP.ISIC96}
M.~Demetriou and M.~M. Polycarpou, ``Fault diagnosis of hyperbolic distributed
  parameter systems,'' in \emph{Proceedings of the 1996 IEEE International
  Symposium on Intelligent Control}.\hskip 1em plus 0.5em minus 0.4em\relax
  IEEE, 1996, pp. 194--199.

\bibitem{Dem.COCV02}
M.~A. Demetriou, ``A model-based fault detection and diagnosis scheme for
  distributed parameter systems: A learning systems approach,'' \emph{ESAIM:
  Control, Optimisation and Calculus of Variations}, vol.~7, pp. 43--67, 2002.

\bibitem{CaiFS.AUT16}
J.~Cai, H.~Ferdowsi, and J.~Sarangapani, ``Model-based fault detection,
  estimation, and prediction for a class of linear distributed parameter
  systems,'' \emph{Automatica}, vol.~66, pp. 122--131, 2016.

\bibitem{FisD.AUT20}
F.~Fischer and J.~Deutscher, ``Flatness-based algebraic fault diagnosis for
  distributed-parameter systems,'' \emph{Automatica}, vol. 117, p. 108987,
  2020.

\bibitem{DeyPM.AUT19}
S.~Dey, H.~E. Perez, and S.~J. Moura, ``Robust fault detection of a class of
  uncertain linear parabolic pdes,'' \emph{Automatica}, vol. 107, pp. 502--510,
  2019.

\bibitem{ZhaY.CSL20}
J.~Zhang, C.~Yuan, W.~Zeng, P.~Stegagno, and C.~Wang, ``Fault detection of a
  class of nonlinear uncertain parabolic pde systems,'' \emph{IEEE Control
  Systems Letters}, vol.~5, no.~4, pp. 1459--1464, 2020.

\bibitem{BanK.ACC12}
A.~Baniamerian and K.~Khorasani, ``Fault detection and isolation of dissipative
  parabolic pdes: Finite-dimensional geometric approach,'' in \emph{2012
  American Control Conference (ACC)}.\hskip 1em plus 0.5em minus 0.4em\relax
  IEEE, 2012, pp. 5894--5899.

\bibitem{CaiJ.ACC16}
J.~Cai and S.~Jagannathan, ``Fault isolation in distributed parameter systems
  modeled by parabolic partial differential equations,'' in \emph{2016 American
  Control Conference (ACC)}.\hskip 1em plus 0.5em minus 0.4em\relax IEEE, 2016,
  pp. 4356--4361.

\bibitem{GhaE.AUT09}
S.~Ghantasala and N.~H. El-Farra, ``Robust actuator fault isolation and
  management in constrained uncertain parabolic pde systems,''
  \emph{Automatica}, vol.~45, no.~10, pp. 2368--2373, 2009.

\bibitem{GhaE.CDC07}
------, ``Detection, isolation and management of actuator faults in parabolic
  pdes under uncertainty and constraints,'' in \emph{2007 46th IEEE Conference
  on Decision and Control}.\hskip 1em plus 0.5em minus 0.4em\relax IEEE, 2007,
  pp. 878--884.

\bibitem{WuL.TNN08}
H.-N. Wu and H.-X. Li, ``A galerkin/neural-network-based design of guaranteed
  cost control for nonlinear distributed parameter systems,'' \emph{IEEE
  transactions on neural networks}, vol.~19, no.~5, pp. 795--807, 2008.

\bibitem{FueCP.JCNN09}
R.~Fuentes, A.~Poznyak, I.~Chairez, and T.~Poznyak, ``Neural numerical modeling
  for uncertain distributed parameter systems,'' in \emph{2009 International
  Joint Conference on Neural Networks}.\hskip 1em plus 0.5em minus 0.4em\relax
  IEEE, 2009, pp. 909--916.

\bibitem{QiL.CES09}
C.~Qi and H.-X. Li, ``Nonlinear dimension reduction based neural modeling for
  distributed parameter processes,'' \emph{Chemical Engineering Science},
  vol.~64, no.~19, pp. 4164--4170, 2009.

\bibitem{ZhaTLJ.TNNLS16}
R.~Zhang, J.~Tao, R.~Lu, and Q.~Jin, ``Decoupled arx and rbf neural network
  modeling using pca and ga optimization for nonlinear distributed parameter
  systems,'' \emph{IEEE transactions on neural networks and learning systems},
  vol.~29, no.~2, pp. 457--469, 2016.

\bibitem{YuaW.SCL11}
C.~Yuan and C.~Wang, ``Persistency of excitation and performance of
  deterministic learning,'' \emph{Systems \& Control Letters}, vol.~60, no.~12,
  pp. 952--959, 2011.

\bibitem{dlt09wChD}
C.~Wang and D.~J. Hill, \emph{Deterministic learning theory for identification,
  recognition, and control}.\hskip 1em plus 0.5em minus 0.4em\relax Boca Raton,
  FL, USA: CRC Press, 2009.

\bibitem{WanH.TNN07}
------, ``Deterministic learning and rapid dynamical pattern recognition,''
  \emph{IEEE Transactions on Neural Networks}, vol.~18, no.~3, pp. 617--630,
  2007.

\bibitem{CheW.IJACSP14}
T.~Chen and C.~Wang, ``Rapid isolation of small oscillation faults via
  deterministic learning,'' \emph{International Journal of Adaptive Control and
  Signal Processing}, vol.~28, no. 3-5, pp. 366--385, 2014.

\bibitem{CheWH.TNNLS13}
T.~Chen, C.~Wang, and D.~J. Hill, ``Rapid oscillation fault detection and
  isolation for distributed systems via deterministic learning,'' \emph{IEEE
  Transactions on Neural Networks and Learning Systems}, vol.~25, no.~6, pp.
  1187--1199, 2013.

\bibitem{ZhaY.ACCESS19}
J.~Zhang, Q.~Gao, C.~Yuan, W.~Zeng, S.-L. Dai, and C.~Wang, ``Similar fault
  isolation of discrete-time nonlinear uncertain systems: An adaptive threshold
  based approach,'' \emph{IEEE Access}, vol.~8, pp. 80\,755--80\,770, 2020.

\bibitem{ZhaYSHW.TCYB19}
J.~Zhang, C.~Yuan, P.~Stegagno, H.~He, and C.~Wang, ``Small fault detection of
  discrete-time nonlinear uncertain systems,'' \emph{IEEE Transactions on
  Cybernetics}, 2019, doi: 10.1109/TCYB.2019.2945629.

\bibitem{LiQ.JPC10}
H.-X. Li and C.~Qi, ``Modeling of distributed parameter systems for
  applications—a synthesized review from time--space separation,''
  \emph{Journal of Process Control}, vol.~20, no.~8, pp. 891--901, 2010.

\bibitem{BenOCW.SIAM17}
P.~Benner, M.~Ohlberger, A.~Cohen, and K.~Willcox, \emph{Model reduction and
  approximation: theory and algorithms}.\hskip 1em plus 0.5em minus 0.4em\relax
  Society for Industrial and Applied Mathematics, 2017.

\bibitem{ChrD.JMAA97}
P.~D. Christofides and P.~Daoutidis, ``Finite-dimensional control of parabolic
  pde systems using approximate inertial manifolds,'' \emph{Journal of
  mathematical analysis and applications}, vol. 216, no.~2, pp. 398--420, 1997.

\bibitem{ArmD.AIChE08}
A.~Armaou and M.~A. Demetriou, ``Robust detection and accommodation of
  incipient component and actuator faults in nonlinear distributed processes,''
  \emph{AIChE journal}, vol.~54, no.~10, pp. 2651--2662, 2008.

\bibitem{tto92pM}
M.~POEWELL, \emph{The Theory of Radial Basis Function Approximation}.\hskip 1em
  plus 0.5em minus 0.4em\relax Oxford: Clarendon Press, 1992.

\bibitem{DonSW.ACCESS19}
X.~Dong, W.~Si, and C.~Wang, ``Global identification of fitzhugh-nagumo
  equation via deterministic learning and interpolation,'' \emph{IEEE Access},
  vol.~7, pp. 107\,334--107\,345, 2019.

\bibitem{DonW.JBC15}
X.~Dong and C.~Wang, ``Identification of the fitzhugh--nagumo model dynamics
  via deterministic learning,'' \emph{International Journal of Bifurcation and
  Chaos}, vol.~25, no.~12, p. 1550159, 2015.

\end{thebibliography}


\begin{thebibliography}{10}
\providecommand{\url}[1]{#1}
\csname url@samestyle\endcsname
\providecommand{\newblock}{\relax}
\providecommand{\bibinfo}[2]{#2}
\providecommand{\BIBentrySTDinterwordspacing}{\spaceskip=0pt\relax}
\providecommand{\BIBentryALTinterwordstretchfactor}{4}
\providecommand{\BIBentryALTinterwordspacing}{\spaceskip=\fontdimen2\font plus
\BIBentryALTinterwordstretchfactor\fontdimen3\font minus
  \fontdimen4\font\relax}
\providecommand{\BIBforeignlanguage}[2]{{%
\expandafter\ifx\csname l@#1\endcsname\relax
\typeout{** WARNING: IEEEtran.bst: No hyphenation pattern has been}%
\typeout{** loaded for the language `#1'. Using the pattern for}%
\typeout{** the default language instead.}%
\else
\language=\csname l@#1\endcsname
\fi
#2}}
\providecommand{\BIBdecl}{\relax}
\BIBdecl

\bibitem{SonHL.TCYB18}
Y.~Song, X.~He, Z.~Liu, W.~He, C.~Sun and F.Y.~Wang, ``Parallel control of distributed parameter systems,'' \emph{IEEE Transactions on Cybernetics}, vol.~48, no.~12, p. 3291--3301, 2018.

\bibitem{MonY.CNSNS12}
G.~Montaseri and M.~J. Yazdanpanah, ``Predictive control of uncertain nonlinear
  parabolic pde systems using a galerkin/neural-network-based model,''
  \emph{Communications in Nonlinear Science and Numerical Simulation}, vol.~17,
  no.~1, pp. 388--404, 2012.
  
\bibitem{XuFY.TCYB21}
K.~Xu, B.~Fan, H.~Yang, L.~Hu and W.~Shen, ``Locally weighted principal component analysis-based multimode modeling for complex distributed parameter systems,'' \emph{IEEE Transactions on Cybernetics}, 2021, doi: 10.1109/TCYB.2021.3061741.

\bibitem{HerM.IS99}
B.~Hernandez-Morales and A.~Mitchell, ``Review of mathematical models of fluid
  flow, heat transfer, and mass transfer in electroslag remelting process,''
  \emph{Ironmaking \& steelmaking}, vol.~26, no.~6, pp. 423--438, 1999.

\bibitem{BouCK.SP15}
S.~Bououden, M.~Chadli, and H.~R. Karimi, ``Control of uncertain highly
  nonlinear biological process based on takagi--sugeno fuzzy models,''
  \emph{Signal Processing}, vol. 108, pp. 195--205, 2015.

\bibitem{ElFG.AIChE07}
N.~H. El-Farra and S.~Ghantasala, ``Actuator fault isolation and
  reconfiguration in transport-reaction processes,'' \emph{AIChE Journal},
  vol.~53, no.~6, pp. 1518--1537, 2007.

\bibitem{LuZH.TII16}
X.~Lu, W.~Zou, and M.~Huang, ``A novel spatiotemporal ls-svm method for complex
  distributed parameter systems with applications to curing thermal process,''
  \emph{IEEE Transactions on Industrial Informatics}, vol.~12, no.~3, pp.
  1156--1165, 2016.

\bibitem{RusCB.SSBM12}
E.~L. Russell, L.~H. Chiang, and R.~D. Braatz, \emph{Data-driven methods for
  fault detection and diagnosis in chemical processes}.\hskip 1em plus 0.5em
  minus 0.4em\relax London, UK: Springer, 2012.

\bibitem{DemA.JRNC12}
M.~Demetriou and A.~Armaou, ``Dynamic online nonlinear robust detection and
  accommodation of incipient component faults for nonlinear dissipative
  distributed processes,'' \emph{International Journal of Robust and Nonlinear
  Control}, vol.~22, no.~1, pp. 3--23, 2012.

%\bibitem{Dem.COCV02}
%M.~A. Demetriou, ``A model-based fault detection and diagnosis scheme for
%  distributed parameter systems: A learning systems approach,'' \emph{ESAIM:
%  Control, Optimisation and Calculus of Variations}, vol.~7, pp. 43--67, 2002.

\bibitem{CaiFS.AUT16}
J.~Cai, H.~Ferdowsi, and J.~Sarangapani, ``Model-based fault detection,
  estimation, and prediction for a class of linear distributed parameter
  systems,'' \emph{Automatica}, vol.~66, pp. 122--131, 2016.

\bibitem{FisD.AUT20}
F.~Fischer and J.~Deutscher, ``Flatness-based algebraic fault diagnosis for
  distributed-parameter systems,'' \emph{Automatica}, vol. 117, p. 108987,
  2020.
  
\bibitem{FenWW.TCYB21}
Y.~Feng, Y.~Wang, B.C.~Wang, and H.X.~Li, ``Spatial decomposition-based fault detection framework for parabolic-distributed parameter processes,'' \emph{IEEE Transactions on Cybernetics}, 2021, doi: 10.1109/TCYB.2021.3049453.

\bibitem{DeyPM.AUT19}
S.~Dey, H.~E. Perez, and S.~J. Moura, ``Robust fault detection of a class of
  uncertain linear parabolic pdes,'' \emph{Automatica}, vol. 107, pp. 502--510,
  2019.

\bibitem{ZhaY.CSL20}
J.~Zhang, C.~Yuan, W.~Zeng, P.~Stegagno, and C.~Wang, ``Fault detection of a
  class of nonlinear uncertain parabolic pde systems,'' \emph{IEEE Control
  Systems Letters}, vol.~5, no.~4, pp. 1459--1464, 2020.

\bibitem{BanK.ACC12}
A.~Baniamerian and K.~Khorasani, ``Fault detection and isolation of dissipative
  parabolic pdes: Finite-dimensional geometric approach,'' in \emph{2012
  American Control Conference (ACC)}.\hskip 1em plus 0.5em minus 0.4em\relax
  IEEE, 2012, pp. 5894--5899.

\bibitem{CaiJ.ACC16}
J.~Cai and S.~Jagannathan, ``Fault isolation in distributed parameter systems
  modeled by parabolic partial differential equations,'' in \emph{2016 American
  Control Conference (ACC)}.\hskip 1em plus 0.5em minus 0.4em\relax IEEE, 2016,
  pp. 4356--4361.

\bibitem{GhaE.AUT09}
S.~Ghantasala and N.~H. El-Farra, ``Robust actuator fault isolation and
  management in constrained uncertain parabolic pde systems,''
  \emph{Automatica}, vol.~45, no.~10, pp. 2368--2373, 2009.

\bibitem{GhaE.CDC07}
------, ``Detection, isolation and management of actuator faults in parabolic
  pdes under uncertainty and constraints,'' in \emph{2007 46th IEEE Conference
  on Decision and Control}.\hskip 1em plus 0.5em minus 0.4em\relax IEEE, 2007,
  pp. 878--884.

\bibitem{WuL.TNN08}
H.-N. Wu and H.-X. Li, ``A galerkin/neural-network-based design of guaranteed
  cost control for nonlinear distributed parameter systems,'' \emph{IEEE
  transactions on neural networks}, vol.~19, no.~5, pp. 795--807, 2008.

\bibitem{QiL.CES09}
C.~Qi and H.-X. Li, ``Nonlinear dimension reduction based neural modeling for
  distributed parameter processes,'' \emph{Chemical Engineering Science},
  vol.~64, no.~19, pp. 4164--4170, 2009.

\bibitem{ZhaTLJ.TNNLS16}
R.~Zhang, J.~Tao, R.~Lu, and Q.~Jin, ``Decoupled arx and rbf neural network
  modeling using pca and ga optimization for nonlinear distributed parameter
  systems,'' \emph{IEEE transactions on neural networks and learning systems},
  vol.~29, no.~2, pp. 457--469, 2016.

\bibitem{YuaW.SCL11}
C.~Yuan and C.~Wang, ``Persistency of excitation and performance of
  deterministic learning,'' \emph{Systems \& Control Letters}, vol.~60, no.~12,
  pp. 952--959, 2011.

\bibitem{dlt09wChD}
C.~Wang and D.~J. Hill, \emph{Deterministic learning theory for identification,
  recognition, and control}.\hskip 1em plus 0.5em minus 0.4em\relax Boca Raton,
  FL, USA: CRC Press, 2009.

\bibitem{WanH.TNN07}
------, ``Deterministic learning and rapid dynamical pattern recognition,''
  \emph{IEEE Transactions on Neural Networks}, vol.~18, no.~3, pp. 617--630,
  2007.

\bibitem{CheW.IJACSP14}
T.~Chen and C.~Wang, ``Rapid isolation of small oscillation faults via
  deterministic learning,'' \emph{International Journal of Adaptive Control and
  Signal Processing}, vol.~28, no. 3-5, pp. 366--385, 2014.

\bibitem{CheWH.TNNLS13}
T.~Chen, C.~Wang, and D.~J. Hill, ``Rapid oscillation fault detection and
  isolation for distributed systems via deterministic learning,'' \emph{IEEE
  Transactions on Neural Networks and Learning Systems}, vol.~25, no.~6, pp.
  1187--1199, 2013.

\bibitem{ZhaY.ACCESS19}
J.~Zhang, Q.~Gao, C.~Yuan, W.~Zeng, S.-L. Dai, and C.~Wang, ``Similar fault
  isolation of discrete-time nonlinear uncertain systems: An adaptive threshold
  based approach,'' \emph{IEEE Access}, vol.~8, pp. 80\,755--80\,770, 2020.

\bibitem{ZhaYSHW.TCYB19}
J.~Zhang, C.~Yuan, P.~Stegagno, H.~He, and C.~Wang, ``Small fault detection of
  discrete-time nonlinear uncertain systems,'' \emph{IEEE Transactions on
  Cybernetics}, vol.~51, no.~2, pp. 750--764, 2019.  

\bibitem{LiQ.JPC10}
H.-X. Li and C.~Qi, ``Modeling of distributed parameter systems for
  applications—a synthesized review from time--space separation,''
  \emph{Journal of Process Control}, vol.~20, no.~8, pp. 891--901, 2010.

\bibitem{BenOCW.SIAM17}
P.~Benner, M.~Ohlberger, A.~Cohen, and K.~Willcox, \emph{Model reduction and
  approximation: theory and algorithms}.\hskip 1em plus 0.5em minus 0.4em\relax
  Society for Industrial and Applied Mathematics, 2017.

\bibitem{ChrD.JMAA97}
P.~D. Christofides and P.~Daoutidis, ``Finite-dimensional control of parabolic
  pde systems using approximate inertial manifolds,'' \emph{Journal of
  mathematical analysis and applications}, vol. 216, no.~2, pp. 398--420, 1997.

\bibitem{ArmD.AIChE08}
A.~Armaou and M.~A. Demetriou, ``Robust detection and accommodation of
  incipient component and actuator faults in nonlinear distributed processes,''
  \emph{AIChE journal}, vol.~54, no.~10, pp. 2651--2662, 2008.

\bibitem{tto92pM}
M.~POEWELL, \emph{The Theory of Radial Basis Function Approximation}.\hskip 1em
  plus 0.5em minus 0.4em\relax Oxford: Clarendon Press, 1992.

%%\bibitem{DonSW.ACCESS19}
%%X.~Dong, W.~Si, and C.~Wang, ``Global identification of fitzhugh-nagumo
%%  equation via deterministic learning and interpolation,'' \emph{IEEE Access},
%%  vol.~7, pp. 107\,334--107\,345, 2019.
%%
%%\bibitem{DonW.JBC15}
%%X.~Dong and C.~Wang, ``Identification of the fitzhugh--nagumo model dynamics
%%  via deterministic learning,'' \emph{International Journal of Bifurcation and
%%  Chaos}, vol.~25, no.~12, p. 1550159, 2015.


%\bibitem{WuWG.TFS16}
%H.-N.~Wu, H.-D.~Wang and L.~Guo, ``Disturbance rejection fuzzy control for nonlinear parabolic PDE systems via multiple observers,'' \emph{IEEE Transactions on Fuzzy Systems}, vol.~24, no.~6, p. 1334--1348, 2016.

%\bibitem{PouA.CES14}
%D.B.~Pourkargar and A.~Armaou, ``Geometric output tracking of nonlinear distributed parameter systems via adaptive model reduction,'' \emph{Chemical Engineering Science}, vol.~116, p. 418--427, 2014.

%\bibitem{ChaC.TFS10}
%Y.-T.~Chang and B.-S.~Chen, ``A fuzzy approach for robust reference-tracking-control design of nonlinear distributed parameter time-delayed systems and its application,'' \emph{IEEE Transactions on fuzzy systems}, vol.~18, no.~6, p. 1041--1057, 2010.

\bibitem{DenLC.TCST05}
H.~Deng, H.X.~Li and G.~Chen, ``Spectral-approximation-based intelligent modeling for distributed thermal processes,'' \emph{IEEE Transactions on Control Systems Technology}, vol.~13, no.~5, p. 686--700, 2005.
\end{thebibliography}
\end{document}